\documentclass[a4paper,11pt]{amsart}
\usepackage{a4wide}

\usepackage[T1]{fontenc}   
\usepackage[latin1]{inputenc}
\usepackage[english]{babel} 
\usepackage{amsmath}
\usepackage{amsfonts}
\usepackage{amssymb}
\usepackage{amsthm}
\usepackage{mathrsfs}
\usepackage{url}
\usepackage{enumerate}
\usepackage{graphicx}
\usepackage[all]{xy}
\usepackage{algorithm}
\usepackage{algorithmic}
\usepackage{color}
\usepackage{hyperref}
\usepackage{booktabs}
\usepackage{marvosym}

\theoremstyle{plain}
\newtheorem{thm}{Theorem}
\newtheorem{theorem}[thm]{Theorem}
\newtheorem{assumption}[thm]{Assumption}{\bfseries}{\itshape}
\newtheorem{fact}{Fact}{\bfseries}{\itshape}
\newtheorem{prop}[thm]{Proposition}
\newtheorem{proposition}[thm]{Proposition}
\newtheorem{problem}{Problem}
\newtheorem{cor}[thm]{Corollary}

\newtheorem{lem}[thm]{Lemma}
\newtheorem{lemma}[thm]{Lemma}
\theoremstyle{definition}
\newtheorem{definition}{Definition}
\newtheorem{example}{Example}
\newtheorem{nota}{Notation}

\theoremstyle{remark}
\newtheorem{rem}{Remark}

\renewcommand{\O}{\mathcal{O}}

\newcommand{\tr}{\ensuremath{\textrm{Tr}}}
\newcommand{\nr}{\ensuremath{\textrm{N}}}
\newcommand{\Z}{\mathbb{Z}}

\newcommand{\fract}[2]{\hbox{\leavevmode
\kern.1em \raise .3ex \hbox{\the\scriptfont0 $#1$}\kern-.1em }\big/
\raise -.5ex\hbox{\kern-.15em \lower .25ex \hbox{\the\scriptfont0 $#2$}}
}

\newcommand{\eqdef}{\stackrel{\text{def}}{=}}

\renewcommand{\leq}{\leqslant} 
 
\renewcommand{\geq}{\geqslant}

\newcommand{\Fq}{\ensuremath{\mathbb{F}_q}}
\newcommand{\Fqq}{\ensuremath{\mathbb{F}_{q^2}}}
\newcommand{\Fqm}{\ensuremath{\mathbb{F}_{q^m}}}

\newcommand{\CC}{\code{C}}
\newcommand{\DC}{\code{D}}

\newcommand{\EC}{\code{E}}
\newcommand{\CCs}{\code{C}^{\sigma}}
\newcommand{\CCse}{\code{C}^{\sigma}_{\textrm{ext}}}
\newcommand{\sqc}[1]{#1^{\star 2}}
\newcommand{\spc}[2]{#1 \star #2}

\newcommand{\cset}[3]{\left\{\left.\left(#1\right)_{#2} ~\right|~#3 \right\}}
\newcommand{\csetstar}[2]{\left\{\left. #1 ~\right|~#2 \right\}}

\newcommand{\mat}[1]{\ensuremath{\boldsymbol{#1}}}
\newcommand{\code}[1]{\ensuremath{\mathscr{#1}}}

\newcommand{\card}[1]{\ensuremath{|#1|}}

\newcommand{\prob}{\mathbf{Prob}}

\newcommand{\Am}{\mat{A}}

\newcommand{\Gm}{\mat{G}}
\newcommand{\Hm}{\mat{H}}

\newcommand{\Gsig}{\mat{G}^{\sigma}}
\newcommand{\xsig}{\xv^{\sigma}}

\newcommand{\Gsige}{\mat{G}^{\sigma}_{\textrm{\bf ext}}}

\newcommand{\Ind}{{\cal{I}}}

\newcommand{\av}{\mat{a}}
\newcommand{\bv}{\mat{b}}
\newcommand{\cv}{\mat{c}}

\newcommand{\ev}{\mat{e}}

\newcommand{\sv}{\mat{s}}
\newcommand{\uv}{\mat{u}}

\newcommand{\vv}{\mat{v}}

\newcommand{\xv}{{\mat{x}}}
\newcommand{\xvzero}{{\mat{x}}_{0}}

\newcommand{\yvzero}{{\mat{y}}_{0}}
\newcommand{\xvone}{{\mat{x}}_{1}}
\newcommand{\xvzo}{{\mat{x}}_{01}}
\newcommand{\xvzeroI}{{\mat{x}}_{\Ind \cup\{0\}}}
\newcommand{\yv}{{\mat{y}}}

\newcommand{\onev}{{\mat{1}}}

\newcommand{\xvpe}{\xv'_{\textrm{\bf ext}}}
\newcommand{\xve}{\xv^{\sigma}_{\textrm{\bf ext}}}

\newcommand{\Lc}{\mathcal{L}}

\newcommand{\F}{\ensuremath{\mathbb{F}}}
\newcommand{\fq}{\F_q}
\newcommand{\fqm}{\F_{q^m}}

\newcommand{\dec}{\mathcal{D}}

\renewcommand{\fq}{\F_{q}}
\renewcommand{\fqm}{\F_{q^m}}

\newcommand{\GRS}[3]{\text{\bf GRS}_{#1}(#2,#3)}

\newcommand{\RS}[2]{\text{\bf RS}_{#1}(#2)}
\newcommand{\Alt}[3]{\code{A}_{#1}(#2,#3)}
\newcommand{\Goppa}[2]{\code{G}\left(#1,#2\right)}

\newcommand{\am}{a^-}

\newcommand{\Coi}[2]{\code{C}_{#1}(#2)}

\newcommand{\polym}[2]{\F_{q^{#2}}[z]_{<#1}}

\newcommand{\sh}[2]{\mathcal{S}_{#2}\left(#1\right)}

\newcommand{\pu}[2]{\mathcal{P}_{#2}\left(#1\right)}
\newcommand{\sC}[1]{\sh{\code{C}}{ #1}}

\newcommand{\loc}[1]{\pi_{#1}}

\renewcommand{\Ind}{\mathcal{I}}

\newcommand{\Jind}{\mathcal{J}}

\begin{document}
\title{Polynomial Time Attack on Wild McEliece Over Quadratic Extensions}

\author{Alain Couvreur}
\address{INRIA \&  LIX, CNRS UMR 7161
--- \'Ecole Polytechnique, 91128 Palaiseau Cedex, France.}
\email{alain.couvreur@lix.polytechnique.fr}
 \author{Ayoub Otmani}
 \address{Normandie Univ, France; 
UR, LITIS, F-76821 Mont-Saint-Aignan, France.}
 \email{ayoub.otmani@univ-rouen.fr}
 \author{Jean--Pierre Tillich}
 \address{INRIA SECRET Project,   
78153 Le Chesnay Cedex, France.}
\email{jean-pierre.tillich@inria.fr}

\maketitle

\begin{abstract}
We present a polynomial-time structural attack against the McEliece system based on Wild Goppa codes defined over a quadratic finite field extension.
We show that such codes can be efficiently distinguished from random codes. The attack uses this property to compute a filtration, that is to say, a family
of nested subcodes which will reveal their secret algebraic description.
\end{abstract}

\bigskip


\bigskip

\bigskip

\section{Introduction}
\label{sec:intro}

\paragraph{\bf The McEliece cryptosystem and its security.}
The McEliece encryption scheme \cite{M78} which dates back to the end of the seventies  
still belongs to the very few public-key cryptosystems which remain unbroken.
It is based on the  famous  Goppa codes family. Several proposals  which suggested to replace 
binary Goppa codes with alternative families did not meet a similar fate.  
They all focus on a specific class of codes equipped with a decoding algorithm:
generalized Reed--Solomon  codes  (GRS for short) \cite{N86} or 
large subcodes of them \cite{BL05}, 
Reed--Muller codes \cite{S94}, algebraic geometry codes \cite{JM96}, LDPC and MDPC codes \cite{BBC08,MTSB13} 
or convolutional codes \cite{LJ12,GSJB14}. Most of them were successfully cryptanalyzed 
\cite{SS92,W10,MS07,FM08,OTD10,CGGOT14,LT13,CMP14,COTG15}.
Each time a description of the underlying 
code suitable for decoding is efficiently obtained. 
But some of them remain unbroken, namely those relying on
MDPC codes \cite{MTSB13} and their cousins \cite{BBC08}, the original binary Goppa codes of
\cite{M78} and their non-binary variants as proposed in \cite{BLP10,BLP11a}. 

Concerning the security of the McEliece proposal based on Goppa codes, weak keys were identified in \cite{G91,LS01} but they can be easily avoided.
There also exist  message recovery attacks 
using generic exponential time decoding algorithms
\cite{LB88,L88,S88,CC98,BLP08,MMT11,BJMM12}. 
More recently,  it  was shown in \cite{FOPT10,GL09}  that the secret structure of Goppa codes can be recovered by an algebraic attack using Gr\"obner bases.
This attack is of exponential nature and is infeasible for the original
McEliece scheme (the number of unknowns is linear in the length of the code),
whereas for variants using Goppa codes with a quasi-dyadic or quasi-cyclic
structure it was feasible due to the huge reduction of the number of unknowns.

\medbreak

\paragraph{\bf Distinguisher for Goppa and Reed-Solomon codes.} 

None of the existing strategies is able to severely dent the security of \cite{M78} when appropriate parameters are taken. Consequently, 
it has even been advocated that the generator matrix of a Goppa 
code does not disclose any visible structure that an attacker could exploit. 
This is strengthened by the fact that Goppa codes
share many characteristics with random codes.
However, in \cite{FGOPT11,FGOPT13}, an algorithm  that manages to 
distinguish between a random code and a high rate Goppa code has been introduced.

\medbreak

\paragraph{\bf Component wise products of codes.}
\cite{MP12} showed that the distinguisher given in \cite{FGOPT11} has an equivalent but simpler description
in terms of component-wise product of codes. 
This product allows in particular to define the square of a code.
This square code operation can be used to distinguish 
a high rate Goppa code from a random one because the dimension of the square of  the dual is 
much  smaller than the one obtained with a random code.
The notion of component-wise product of codes was first put forward to unify many 
different algebraic decoding algorithms \cite{P92,K92},
 then exploited in cryptology in \cite{W10} to break a McEliece  variant based on random subcodes
of GRS codes \cite{BL05} and  in \cite{CMP11,CMP11a, CMP14, CMP14a}
to study the security of encryption schemes using algebraic geometry codes.
Component-wise powers of codes are also studied in the context of
secret sharing and secure multi-party computation  \cite{CCCX09,CCX11}.

\medbreak

\paragraph{\bf Filtration key-recovery attacks.}
The works \cite{FGOPT11,FGOPT13}, without undermining the security of \cite{M78},
prompts to wonder whether it would be possible to devise an attack exploiting the distinguisher. 
That was indeed the case in \cite{CGGOT14} for McEliece-like public-key encryption
schemes relying on modified GRS codes \cite{BL12a,BBCRS11,W06}.
Additionally, \cite{CGGOT14} has shown that the unusually low dimension of
the square code of a generalized GRS code enables to compute a {\em filtration},
that is a nested sequence of subcodes, allowing the recovery of its algebraic
structure. This gives an attack that is radically different from the
Sidelnikov-Shestakov approach \cite{SS92}.
Notice that the first step of the Sidelnikov-Shestakov attack which consists in computing the minimal codewords 
and then using this information for recovering the algebraic structure has been fruitful for breaking 
other families of codes: for instance binary Reed-Muller codes 
\cite{MS07} or low-genus algebraic geometry codes \cite{FM08}.
This is not the approach we have followed here, because finding such codewords seems 
out of reach for the codes we are interested in, namely Goppa codes.
Our filtration attack is really a new paradigm for breaking public key cryptosystems
based on algebraic codes which in particular avoids 
the possibly very expensive computation of minimum weight codewords.

\medbreak

\paragraph{\bf Our contribution.}

The purpose of this article is to show that the
filtration attack of \cite{CGGOT14}
which gave a new way of attacking 
a McEliece scheme based on GRS codes
can be generalized to other families of codes.
Notice that this filtration approach was also followed later on with great success to break all schemes 
based on algebraic geometry codes \cite{CMP14}, whereas the aforementioned 
attack of Faure and Minder \cite{FM08} could handle only the case
of very low genus curves, due precisely to the expensive computation of the 
minimal codewords.
A tantalizing project would be to attack Goppa code based McEliece schemes,
or more
generally alternant code based schemes.
The latter family of codes are subfield subcodes defined over some field $\fq$
of GRS 
codes defined over a field extension $\F_{q^m}$.
Even for the smallest possible
field extension, that is for $m=2$, the cryptanalysis
of alternant codes
 is a completely open question.
Codes of this kind have indeed been
proposed  as possible improvements of  the original McEliece scheme, under
the form of {\it wild Goppa
codes} in \cite{BLP10}.
These 
are Goppa codes associated to polynomials of the form $\gamma^{q-1}$
where $\gamma$ is irreducible.
Notice that all irreducible binary Goppa codes of the original McEliece system
are actually wild Goppa codes.
Interestingly enough, it turns out that these wild Goppa codes for $m=2$ can be distinguished from random codes for a very large range of parameters
by observing that the square code 
of some of their shortenings
 have a small dimension compared to squares of random codes of the same
dimension.
It should be pointed out that
in the propositions of \cite{BLP10},
the case $m=2$ was particularly attractive since
it provided the smallest key sizes.

We show here that this distinguishing property can be used to compute an interesting
filtration of the public code, that is to say a family of nested subcodes of
the public Goppa code such that each element of the family is an alternant code
with the same support. 
This filtration can in turn be used to recover the algebraic description
of the Goppa code as an alternant code, which yields
an efficient key recovery attack.
This attack has been implemented in Magma \cite{BCP97} and allowed to break completely all the schemes 
with a claimed 128 bit security in Table 7.1 of  \cite{BLP10} corresponding to $m=2$ when the
degree of $\gamma$ is larger than $3$. 
This corresponds precisely to the case where these codes can be distinguished from random codes
by square code considerations. The filtration attack has a polynomial time complexity and basically boils
down to linear algebra. This is the first time in the 35 years of existence of the McEliece scheme based on Goppa codes that a polynomial time attack has been found on 
it. It questions the common belief that GRS codes are weak for a cryptographic use while Goppa codes are
secure as soon as $m \geq 2$ and that for the latter,
only generic information-set-decoding attacks apply.
 It also  raises the issue whether this algebraic distinguisher of Goppa
 and more generally alternant codes (see \cite{FGOPT13}) based on square code considerations
can be turned into an attack in the other cases where it applies (for instance
for Goppa codes of rate close enough to $1$). 
Finally, it is worth pointing out that our attack works
against codes without external symmetries confirming that the mere appearance of randomness is far from being
enough to defend codes against algebraic attacks.

It should also be  pointed out that subsequently to this work, it has been shown in \cite{FPP14} that 
one of the parameters of \cite{BLP10} that we have broken here can be attacked by Gr\"obner basis techniques 
by introducing an improvement of the
algebraic modelling of \cite{FOPT10} together with a new way of exploiting non-prime fields $\fq$. 
Unlike our attack, it also applies to field extensions that are larger than $2$ and to variations of the wild Goppa codes, called ``wild Goppa incognito'' in
\cite{BLP11a}.  However this new attack is exponential in nature and can be thwarted by using either more conservative 
parameters or prime fields $\fq$. Our attack on the other hand is much harder to avoid due to its polynomial complexity
and choosing $m=2$ together with wild Goppa codes seems now something that has to be considered with great care in a McEliece 
cryptosystem after our work.

\medbreak

\paragraph{\bf Outline of the article}
Our objective is to provide a self-contained article which could
be read by cryptographers who are not aware with coding theory.
For this reason, notation and many classical prerequisites
are given in Section~\ref{sec:nota}. The core of our attack is presented
in Sections~\ref{sec:distinguisher} and~\ref{sec:nested}.
Section~\ref{sec:distinguisher} presents a distinguisher on the public key
i.e. a manner to distinguish such codes from random ones and
Section~\ref{sec:nested} uses this distinguisher to compute a family of
nested subcodes of the public key providing information on the secret key.
Section~\ref{sec:attack} is devoted to a short overview of
the last part of the attack. Further technical details on the attack are given
in Appendix \ref{sec:in_depth}. Section \ref{sec:ifcompjust} is devoted
to the limits and possible extensions of this attack and
Section~\ref{sec:example} discusses the theoretical complexity of our 
algorithm and presents several running times of our Magma implementation
of the attack.

\medbreak

\paragraph{\bf Note} The material of this article was presented at the conference
{\em EUROCRYPT~2014} (Copenhagen, Denmark) and published in its proceedings \cite{COT14}.
 Due to space constraints, most of the proofs were
 omitted in the proceedings
version. The present article is a long revisited
version including all the missing proofs. 
Any proof which we did not consider as fundamental 
has 
been sent to the appendices.
We encourage the reader first to read the article without these
proofs and then read the appendices.

\section{Notation, Definitions and Prerequisites}
  \label{sec:nota}

We introduce in this section notation we will use in the sequel. We assume that the reader is familiar with notions from coding theory.
We refer to \cite{MS86} for the terminology. 

\subsection{Vectors, matrices  and Schur product}\label{ss:notaVect}
Vectors and  matrices are respectively denoted in bold letters 
and bold capital letters such as $\av$ and $\Am$.
We always denote the entries of a vector $\uv \in \Fq^n$ by $u_0, \ldots , u_{n-1}$.
Given a subset $\Ind \subset \{0, \ldots , n-1\}$, we denote by
$\uv_{\Ind}$ the vector $\uv$ \emph{punctured} at $\Ind$, that is to say, \emph{every entry with index in $\Ind$ is removed}.
When $\Ind = \{j\}$ 
we allow ourselves to write $\uv_{j}$ instead of $\uv_{\{j\}}$.
The component-wise product also called the Schur product $\spc{\uv}{\vv}$ of two vectors $\uv, \vv \in \Fq^n$ is defined as:
$$\spc{\uv}{\vv} \eqdef (u_0 v_0, \ldots , u_{n-1}v_{n-1}).$$
The $i$--th power $\uv\star \cdots \star \uv$ is denoted by
$\uv^{i}$. When every entry $u_i$ of $\uv$ is nonzero,
we set $$\uv^{-1} \eqdef (u_0^{-1}, \ldots , u_{n-1}^{-1}),$$ 
and more generally  for all $i$, we define $\uv^{-i}$ in the same manner. 
The operation $\star$ has an identity element, which is nothing but the
all-ones vector $(1, \ldots, 1)$ denoted by $\onev$.

\subsection{Polynomials}\label{subsec:poly}
The ring of polynomials with coefficients in $\Fq$  is denoted by $\Fq[z]$, 
while the subspace of $\Fq[z]$ of polynomials of degree
less than $t$ is denoted by $\Fq[z]_{<t}$.
For every rational fraction $P \in \Fq(z)$, 
with no poles at the elements $u_0, \ldots, u_{n-1}$,
$P(\uv)$ stands for $\left(P(u_0),
\ldots , P(u_{n-1}) \right)$. 
In particular for all $a , b\in\Fq$, $a\uv+b$ is the vector $(au_0+b, \ldots ,
au_{n-1}+b)$.

The \emph{norm} and \emph{trace} from $\Fqm$ to $\Fq$ can be viewed as
polynomials and applied componentwise to vectors in $\Fqm^n$.
In the present article we focus in particular on quadratic extensions ($ m = 2$)
which motivates the following notation
for all $\xv \in \Fqq^n$:
$$
   \nr(\xv)  \eqdef   \left(x_0^{q+1}, \ldots , x_{n-1}^{q+1} \right), \qquad
   \tr (\xv)  \eqdef  
  \left(x_0^q + x_0, \ldots , x_{n-1}^q + x_{n-1} \right).
$$
Finally, to each vector $\xv = (x_0, \ldots , x_{n-1}) \in \Fq^n$, we
associate its {\em locator polynomial} denoted as $\loc{\xv}$ and defined as:
$$
\loc{\xv}(z)\eqdef \prod_{i=0}^{n-1}(z-x_i).
$$

\subsection{Operations on codes}

For a given code $\DC \subseteq \Fq^n$ and a subset $\Ind \subseteq \{0, \ldots, n-1\}$
 the \emph{punctured} code $\pu{\DC}{\Ind}$ and \emph{shortened} code  $\sh{\DC}{\Ind}$
are defined as:
$$
\begin{aligned}
\pu{\DC}{\Ind} &\eqdef \Big \{ {(c_i)}_{i \notin \Ind}~|~ \cv \in \DC \Big \}; \\
\sh{\DC}{\Ind} &\eqdef \Big \{ {(c_i)}_{i \notin \Ind}~|~\exists \cv={(c_i)}_i \in \DC \text{ such that } 
\forall i \in \Ind,\; c_i = 0 \Big \}.  
\end{aligned}
$$
Instead of writing $\pu{\DC}{\{j\}}$ and $\sh{\DC}{\{j\}}$ when $\Ind = \{j\}$ we rather use the notation
$\pu{\DC}{j}$ and $\sh{\DC}{j}$.
The following classical result will be used repeatedly.

 \begin{proposition}
 \label{pr:dualshort}
 Let $\code{A} \subseteq \Fq^n$ be a code and $\Ind \subseteq \{0, \ldots , n-1\}$
   be a set of positions.  Then,
   $$
 \sh{\code{A}}{\Ind}^{\perp} = \pu{\code{A}^{\perp}}{\Ind}
   \quad \textrm{and}\quad
   \pu{\code{A}}{\Ind}^{\perp} = \sh{\code{A}^{\perp}}{\Ind}.
   $$
 \end{proposition}

 \begin{proof}
   See for instance \cite[Theorem 1.5.7]{HP03}
 \end{proof}

Given a code $\CC$ of length $n$ over a finite field extension $\F_{q^m}$ of
$\F_q$, the {\em subfield subcode} of $\CC$ over $\Fq$ is the code
$\CC \cap \Fq^n$. The {\em trace code} $\tr (\CC)$ is the image
of $\CC$ by the componentwise trace map $\tr_{\F_{q^m}/ \Fq}$.
 We recall an important result due to Delsarte establishing a link between 
subfield subcodes and trace codes.

 \begin{theorem}[Delsarte Theorem {\cite[Theorem 2]{D75}}]
 \label{th:Delsarte}
 Let $\code{E}$ be a linear code of length $n$ defined over $\Fqm$.
 Then 
 $$
 (\code{E} \cap \Fq^n) = {\tr(\code{E}^\perp)}^\perp.
 $$
 \end{theorem}

The following classical result is extremely useful in the next sections.

\begin{proposition}\label{prop:commutation}
  Let $\code{A}$ be a code over $\Fqm$ of length $n$ and $\Ind \subseteq
  \{0, \ldots , n-1\}$, then, we have:
  \begin{enumerate}[(a)]
  \item\label{it:punctTr} $ \pu{\tr(\code{A})}{\Ind}  =
        \tr (\pu{\code{A}}{\Ind}) $;
      
        \item\label{it:shortTr}  $\tr(\sh{\code{A}}{\Ind})  \subseteq
    \sh{\tr (\code{A})}{\Ind}$.
  
  \item\label{it:shortsub} $\sh{\code{A}}{\Ind} \cap \Fq^{n - |\Ind|}  =
        \sh{\code{A} \cap \Fq^n}{\Ind} $;
  
  \item\label{it:punctsub} $\pu{\code{A}\cap \Fq^n}{\Ind}  \subseteq
    \pu{\code{A}}{\Ind} \cap \Fq^{n-|\Ind|}$;
      \end{enumerate}
\end{proposition}

\begin{proof}
The componentwise trace map and the puncturing map commute 
with each 
other, which proves (\ref{it:punctTr}). 
To prove (\ref{it:shortTr}), let $\cv \in \code{A}$ be a codeword
whose entries with indexes in $\Ind$ are all equal to $0$. 
We have $\tr(c_{\Ind}) \in \tr(\sh{\code{A}}{\Ind})$.
Moreover, the entries
of $\tr (\cv)$ with indexes in $\Ind$
are also equal to $0$, hence $\tr(\cv_{\Ind}) =
\tr(\cv)_{\Ind} \in \sh{\tr(\code{A})}{\Ind}$. This proves (\ref{it:shortTr}).
By duality, (\ref{it:shortsub}) and 
(\ref{it:punctsub}) can be directly deduced from
(\ref{it:punctTr}) and (\ref{it:shortTr}) thanks to
Proposition~\ref{pr:dualshort} and Theorem~\ref{th:Delsarte}.
\end{proof}

\subsection{Generalized Reed--Solomon and Alternant codes}

\begin{definition}[Generalized Reed-Solomon code] \label{def:defGRS}
Let $q$ be a prime power and
$k, n$ be integers such that $1 \leqslant k < n \leqslant q$.
Let $\xv$ and $\yv$ be two $n$-tuples such that the entries of $\xv$ are
pairwise distinct elements of $\fq$  and those of $\yv$ are nonzero elements in $\fq$.
The {\em generalized Reed-Solomon} code (GRS in short) $\GRS{k}{\xv}{\yv}$ of dimension $k$ associated to 
$(\xv,\yv)$ is defined as
\begin{eqnarray*}
\GRS{k}{\xv}{\yv} &\eqdef& \Big
\{ \big( y_0p(x_0),\dots{},y_{n-1}p(x_{n-1}) \big) ~\big |~ p \in \fq[z]_{< k} 
\Big\} \\
&=&  \Big
\{  \spc{\yv}{p(\xv)} ~\big |~ p \in \fq[z]_{< k} 
\Big\}.
\end{eqnarray*}

Reed-Solomon codes correspond to the case
where $\yv = \onev$ and are denoted as $\RS{k}{\xv}$.
The vectors $\xv$ and $\yv$ are called the {\em support} and the {\em multiplier} of the code.
\end{definition}

In the sequel, we will also use the terms support and multiplier without refering to a 
generalized Reed-Solomon code -- this term will also appear in the context of Goppa and
alternant codes. In this case, when we say that a vector $\xv \in \fq^n$ is a {\em support}, this means
that all its entries are distinct. Likewise, when we say that a vector $\yv \in \fq^n$ is a {\em multiplier},
this means that all its entries are different from zero. 
\begin{proposition}
  \label{prop:dualGRS} Let $\xv, \yv$ be as in 
  Definition \ref{def:defGRS}.
  Then,
  $$
\GRS{k}{\xv}{\yv}^{\perp} = \GRS{n-k}{\xv}{\spc{\yv^{-1}}{\loc{\xv}'(\xv)^{-1}}}
$$
where $\loc{\xv}$ is the locator polynomial of $\xv$ as defined in
\S~\ref{subsec:poly} and $\loc{\xv}'$ denotes its first derivative.
  \end{proposition}

  \begin{proof}
    See for instance \cite[Prop. 5.2 \& Problems 5.6, 5.7]{R06}.
  \end{proof}

This leads to the definition of alternant codes (\cite[Chap. 12, \S~2]{MS86}).
\begin{definition}[Alternant code]\label{def:subfield_subcode}
Let $\xv, \yv \in \Fqm^n$ be a support and a multiplier.
Let $\ell$ be a positive integer,
the alternant code $\Alt{\ell}{\xv}{\yv}$ defined over $\fq$
is defined as
$$
\Alt{\ell}{\xv}{\yv} \eqdef \GRS{\ell}{\xv}{\yv}^{\perp} \cap \fq^n.
$$
The integer $\ell$ is referred to as the {\em degree} of the alternant code
and $m$ as its {\em extension degree}.
\end{definition}

\begin{prop}[{\cite[Chap. 12, {\S}~2]{MS86}}]\label{prop:designed}
  Let $\xv, \yv$ be as in Definition \ref{def:subfield_subcode}. 
  \begin{enumerate}
  \item\label{it:designed_dim} $\dim_{\Fq} \Alt{\ell}{\xv}{\yv} \geq n - m\ell$;
  \item $d_{\textrm{min}}(\Alt{\ell}{\xv}{\yv}) \geq \ell+1$;
  \end{enumerate}
  where $d_{\textrm{min}}(\cdot)$ denotes the minimum distance of a code.
\end{prop}

From Definition~\ref{def:subfield_subcode}, it is clear that alternant codes inherit the decoding algorithms of 
the underlying GRS codes.
The key feature of an alternant code is the following fact (see \cite[Chap. 12, {\S}~9]{MS86}):
\begin{fact}
\label{fa:decoding}
There exists a polynomial time algorithm decoding all errors of Hamming weight at 
most $\lfloor \frac{\ell}{2} \rfloor$ once
the vectors $\xv$ and $\yv$ are known.
\end{fact}

The following description of alternant codes, will be extremely useful in
this article.

\begin{lemma}
  \label{lem:decription_as_eval} Let $\xv$, $\yv$, $\ell$ be as in Definition~\ref{def:subfield_subcode}. We have:
\begin{eqnarray*}
\Alt{\ell}{\xv}{\yv} &=&  \cset{\frac{1}{y_i \loc{\xv}'(x_i)}f(x_i)}{0\leq i <n}{f\in \polym{n-\ell}{m}} \cap \fq^n\\
&=& \csetstar{ \yv^{-1} \star \loc{\xv}'(\xv)^{-1} \star f(\xv) }
{f\in \polym{n-\ell}{m}} \cap \fq^n.
\end{eqnarray*}
\end{lemma}

\subsection{Classical Goppa codes}

\begin{definition}
\label{def:ClassicalGoppa}
Let $\xv \in \Fqm^n$ be a support 
and $\Gamma \in \Fqm[z]$
be a polynomial such that $\Gamma (x_i)\neq 0$ for
all $i \in \{0, \ldots, n-1\}$.
The classical Goppa code $\Goppa{\xv}{\Gamma}$ over $\Fq$
associated to $\Gamma$ and supported by $\xv$
is defined as
$$
\Goppa{\xv}{\Gamma} \eqdef \Alt{\deg \Gamma}{\xv}{\Gamma(\xv)^{-1}}.
$$
We call $\Gamma$ the {\em Goppa polynomial} and $m$
the {\em extension degree} of the Goppa code.
\end{definition}
As for alternant codes, the following description of Goppa codes,
which is due to Lemma~\ref{lem:decription_as_eval}
will be extremely useful in this article.
\begin{lemma}\label{lem:descr_Goppa_as_eval}  
Let $\xv, \Gamma$ be as in Definition \ref{def:ClassicalGoppa}. We have,
  $$
  \begin{aligned}
      \Goppa{\xv}{\Gamma}&=
  \cset{\frac{\Gamma(x_i)}{\loc{\xv}'(x_i)}f(x_i)}{0\leq i <n}{f \in \polym{n-\deg (\Gamma)}{m}}\cap \Fq^{n}\\
                         &=
  \csetstar{\Gamma(\xv)\star \loc{\xv}'(\xv)^{-1}\star f(\xv) }
{f \in \polym{n-\deg (\Gamma)}{m}}\cap \Fq^{n}.
  \end{aligned}
  $$
\end{lemma}

The interesting point about this subfamily of alternant codes is that under some conditions, Goppa codes can correct more errors than a general alternant code. 
\begin{theorem}[{\cite[Theorem 4]{SKHN76}}] \label{thm:SKHN76}
Let $\gamma \in \Fqm[z]$ be a squarefree polynomial.
Let $\xv \in \Fqm^n$ be a support,
then
$$
\Goppa{\xv}{\gamma^{q-1}} = \Goppa{\xv}{\gamma^q}.
$$
\end{theorem}

Codes with such a Goppa polyomial are called {\em wild Goppa codes}.
From Fact \ref{fa:decoding}, wild Goppa codes correct up to
$\lfloor \frac{qr}{2}\rfloor$ errors in polynomial-time instead
of just $\lfloor \frac{(q-1)r}{2}\rfloor$ if 
viewed as
$\Alt{r(q-1)}{\xv}{\gamma^{-(q-1)}(\xv))}$.
On the other hand, these
codes have dimension $\geq n-mr(q-1)$  instead of $\geq n-mrq$. 
Notice that when $q=2$, this amounts to double the error correction capacity.
It is one of the reasons why binary Goppa codes have been chosen in the original McEliece scheme or why Goppa codes with Goppa polynomials of the form $\gamma^ {q-1}$  are proposed in \cite{BLP10,BLP11a}.

\begin{rem}
  \label{rem:wild_Goppa}
  Actually, {\cite[Theorem 4]{SKHN76}} is more general and asserts that
  given irreducible polynomials $f_1, \ldots, f_s \in \Fqm[z]$,
  a polynomial $g$ prime to $f_1\cdots f_s$
  and positive integers $a_1, \ldots, a_s$, then
  $$
  \Goppa{\xv}{f_1^{a_1 q - 1} \cdots f_s^{a_s q - 1} g}
    = \Goppa{\xv}{f_1^{a_1 q} \cdots f_s^{a_s q} g}.
  $$
  
\end{rem}

\subsection{Shortening Alternant and Goppa codes}\label{ss:shorten_alt}
The shortening operation will play a crucial role in our attack.
For this reason, we recall the following classical result. We give
a proof because of a lack of references.

\begin{proposition}\label{pr:shortened_alternant_code}
Let $\xv \in \fqm^n$ be a support and
let $\yv \in \fqm^n$ be a multiplier, then
$$
\sh{\Alt{r}{\xv}{\yv}}{\Ind} = 
    \Alt{r}{\xv_{\Ind}}{\yv_{\Ind}}.
$$
\end{proposition}

\begin{proof}
This proposition follows on the spot from the definition of the
alternant code $\Alt{r}{\xv}{\yv}$: there is a parity-check $\Hm$ for it 
with entries over $\fqm$ which is the generating matrix of $\GRS{r}{\xv}{\yv}$.
A parity-check matrix of the shortened code
$\sh{\Alt{r}{\xv}{\yv}}{\Ind}$ is obtained by throwing away the columns of $\Hm$
that belong to $\Ind$. 
That is to say, by puncturing $\GRS{r}{\xv}{\yv}$
at $\Ind$.
This parity-check matrix is therefore the generator matrix of
$\GRS{r}{\xv_{\Ind}}{\yv_{\Ind}}$ and 
the associated code is $ \Alt{r}{\xv_{\Ind}}{\yv_{\Ind}}$.
\end{proof}

\begin{cor}\label{cor:shortened_Goppa}
  Let $\Gamma \in \Fqm [z]$ and $\xv \in \Fqm^n$ be a support.
  Let $\Ind \subseteq \{0, \ldots, n-1\}$, then
  $$
  \sh{\Goppa{\xv}{\Gamma}}{\Ind} = \Goppa{\xv_{\Ind}}{\Gamma}.
  $$
\end{cor}

\subsection{McEliece encryption scheme}

We recall here the general principle of McEliece public-key 
scheme \cite{M78}. 
The key generation algorithm picks a random $k \times n$ generator matrix 
$\Gm$ of a code $\CC$ over $\Fq$ 
which is itself randomly picked in a family of codes for
which $t$ errors can be efficiently corrected.
The \emph{secret} key is the decoding algorithm $\dec$ associated to $\CC$ and 
the \emph{public} key is $\Gm$. To encrypt $\uv \in \Fq^k$, the sender chooses
 a random vector $\ev$ in $\Fq^n$ of Hamming weight 
less than or equal to $t$ and computes the ciphertext 
$\cv = \uv \Gm + \ev$.
The receiver then recovers the plaintext by applying $\dec$ on $\cv$.

McEliece based his scheme solely on binary Goppa codes. In \cite{BLP10,BLP11a}, it is 
advocated to use $q$-ary wild Goppa codes,
\textit{i.e.} codes with Goppa polynomials of the form
$\gamma^{q-1}$ 
because of their better error correction capability (Theorem~\ref{thm:SKHN76}). 
In this paper, we precisely focus on 
these codes but defined over quadratic extensions ($m = 2$). We shall see how it is 
possible to fully recover their secret structure
under some mild condition on $q$ and the degree of $\gamma$
(further details in Table \ref{tab:limit}).

\section{A Distinguisher based on Square Codes}\label{sec:distinguisher}

\subsection{Square code}
One of the keys for the distinguisher presented here and the attack outlined in
the subsequent sections is a 
special property of certain alternant codes with respect to the component-wise
product.
 \begin{definition}[Product of codes, square code]
Let $\code{A}$ and $\code{B}$ be two codes of length $n$. The
\emph{Schur product code} denoted by $\spc{\code{A}}{\code{B}}$ is the vector space
spanned by all products $\spc{\av}{\bv}$ for all $(\av,\bv)\in \code{A}\times \code{B}$.
When $\code{B} = \code{A}$, 
$\spc{\code{A}}{\code{A}}$ is called the \emph{square code} of $\code{A}$
and is denoted by ${\code{A}}^{\star 2}$.
\end{definition}

The dimension of the Schur product 
 is easily bounded by:
\begin{proposition}\label{pr:typ_dim}
Let $\code{A}$ and  $\code{B}$ be two linear codes  $\subseteq \Fq^n$, then
\begin{eqnarray}
\label{eq:asymprod} \dim\left(\spc{\code{A}}{\code{B}}\right) & \leq &\min \left\{ n,\
\dim{\code{A}} \dim{\code{B}} - {\dim (\code{A}\cap\code{B}) \choose 2}\right\} \\
\label{eq:square} \dim  \left({\code{A}}^{\star 2}\right) & \leq & \min \left\{n,\ \binom{\dim(\code{A})+1}{2}\right\}.
 \end{eqnarray}
\end{proposition}

\begin{proof}
Let $\{e_1, \dots , e_s\}$ be a basis of $\code{A}\cap \code{B}$.
Complete it as two bases
${B}_{\code{A}} = \{e_1, \dots , e_s, a_{s+1},$ $\dots , a_k\}$
and $B_{\code{B}}= \{e_1, \dots, e_s, b_{s+1}, \dots, b_{\ell}\}$ of $\code{A}$
and $\code{B}$ respectively.
The Schur products $\spc{\uv}{\vv}$ where $\uv \in B_{\code{A}}$
and $\vv \in B_{\code{B}}$ span $\spc{\code{A}}{\code{B}}$.
The number of such products is $k\ell = \dim \code{A} \dim \code{B}$
minus the number of products which are counted twice, namely the products
$\spc{e_i}{e_j}$ with $i\neq j$ and their number is precisely
${s \choose 2}$. This proves~(\ref{eq:asymprod}).
The inequality given in~(\ref{eq:square}) is a consequence of~(\ref{eq:asymprod}).
\end{proof}

It is proved in \cite{CCMZ15, R15a} that, almost all codes of a given length and dimension 
reach these bounds while GRS codes behave completely differently when they have
the same support.

\begin{proposition}\label{prop:GRSsquare}
Let $\xv \in \Fq^n$ be a support 
and $\yv$, $\yv'$ be two multipliers in $\Fq^n$.
Then,
\begin{enumerate}[(i)]
\item $\spc{\GRS{k}{\xv}{\yv}}{\GRS{k'}{\xv}{\yv'}} = \GRS{k+k'-1}{\xv}{\yv\star \yv'}$; \label{cas1:grssquare}
\item ${\GRS{k}{\xv}{\yv}}^{\star 2}=\GRS{2k-1}{\xv}{\yv\star\yv}$. \label{cas2:grssquare}
\end{enumerate}
\end{proposition}

This proposition shows that the dimension of
$\spc{\GRS{k}{\xv}{\yv}}{\GRS{k'}{\xv}{\yv'}}$ does not scale multiplicatively as
$kk'$ but additively as $k+k'-1$.
 It has been used the first time in cryptanalysis in \cite{W10} and appears 
for instance explicitly as Proposition 10 in  \cite{MMP11a}. We provide the proof here because it 
is crucial for understanding why the Schur products of GRS codes and some alternant codes behave in 
a non generic way.

\begin{proof}[Proof of Proposition~\ref{prop:GRSsquare}]
 In order to prove (\ref{cas1:grssquare}), let $\cv=(y_0 f(x_0),\dots,y_{n-1}f(x_{n-1})) \in \GRS{k}{\xv}{\yv}$ and
 $\cv'=(y'_0 g(x_0),\dots,y'_{n-1}g(x_{n-1}))$ $\in\GRS{k'}{\xv}{\yv'}$ where
 $\deg(f)\leq k-1$ and  $\deg(g)\leq k'-1$.
 Then $\spc{\cv}{\cv'}$ is of the form: 
 \begin{align*}
 \spc{\cv}{\cv'} & = (y_0y_0' f(x_0)g(x_0),\dots,y_{n-1}y'_{n-1} f(x_{n-1})g(x_{n-1}))\\
                 & =(y_0y_0'r(x_0),\dots,y_{n-1}y'_{n-1}r(x_{n-1}))
 \end{align*}
 where $\deg(r)\leq k+k'-2$.
 Conversely, any element $(y_0y_0'r(x_0),\dots,y_{n-1}y'_{n-1}r(x_{n-1}))$
 where $\deg(r)\leq k+k'-2$, is a linear combination of Schur
 products of two elements of $\GRS{k}{\xv}{\yv}$.
Statement (\ref{cas2:grssquare}) is a consequence of 
(\ref{cas1:grssquare}) by putting $\yv'=\yv$ and $k'=k$.
\end{proof}

Since an alternant code is a subfield subcode of a GRS code,
we might suspect that products of alternant codes have also
low dimension compared to products of random codes.
This is  true but in a very attenuated form  as shown by:

\begin{theorem}\label{thm:sq_alternant}
Let $\xv \in \Fqm^n$ be a support and
$\yv, \yv' \in \Fqm^n$ be two multipliers.
Then, 
\begin{equation}\label{eq:starAlt}
\spc{\Alt{s}{\xv}{\yv}}{\Alt{s'}{\xv}{\yv'}} \subseteq \Alt{s+s'-n+1}{\xv}{\yv''},
\end{equation}
for $\yv'' \eqdef \yv \star \yv' \star \loc{\xv}'(\xv)$.
\end{theorem}

\begin{proof}
Let $\cv, \cv'$ be respective elements of $\Alt{s}{\xv}{\yv}$
and $\Alt{s'}{\xv}{\yv'}$. 
From Lemma~\ref{lem:decription_as_eval},
$$
\cv = f(\xv) \star \yv^{-1} \star \loc{\xv}'(\xv)^{-1}
\quad {\rm and} \quad \cv' = g(\xv) \star {\yv'}^{-1} \star \loc{\xv}'(\xv)^{-1}
$$ 
for some polynomials $f$ and $g$ of degree $<n-s$ and $<n-s'$ respectively.
This implies that
$$\spc{\cv}{\cv'} = h(\xv) \star \yv^{-1} \star {\yv'}^{-1} \star
\loc{\xv}'(\xv)^{-2}$$
where $h\eqdef fg$ is a polynomial of degree $< 2n -(s+s')-1$.
Moreover, since $\cv$ and $\cv'$ have their entries in $\Fq$ then so has
$\spc{\cv}{\cv'}$. Consequently,
$$\spc{\cv}{\cv'} \in  \GRS{ 2n -(s+s')-1}{\xv}{\yv^{-1} \star {\yv^\prime}^{-1} \star {\loc{\xv}^\prime}(\xv)^{-2}}\cap \Fq^n.$$
From Definition \ref{def:subfield_subcode}, the above code equals $\Alt{s+s'-n+1}{\xv}{\yv''}$ for 
$\yv'' = \yv \star \yv' \star \loc{\xv}'(\xv)$.
\end{proof}

\subsection{The particular case of wild Goppa codes
over quadratic extensions}\label{ss:insight}

Theorem~\ref{thm:sq_alternant}
generalizes Proposition~\ref{prop:GRSsquare} which
 corresponds to the particular case where the
extension degree $m$ is equal to $1$. 
 However, when $m>1$, the right hand term
 of (\ref{eq:starAlt})
 is in general the full space $\Fq^n$.
 Indeed, assume that $m>1$ and that the dimensions of
 $\Alt{s}{\xv}{\yv}$ and $\Alt{s'}{\xv}{\yv'}$ are equal to
 $n- sm$ and $n-s'm$ respectively.
 If we assume that both codes have non trivial dimensions then we should have 
 $n-sm >0$ and $n-s'm>0$ which implies that 
 $s , s' < n/m \leq n/2$ and hence:
 $$(s+s')-n+2 \leq 0$$
 which entails that $\Alt{s+s'-n+1}{\xv}{\yv''}$ is the full space $\Fq^n$.

However, in the case  $m=2$ and when either:
\begin{enumerate}[(i)]
\item
$\Alt{s}{\xv}{\yv}$  or $\Alt{s'}{\xv}{\yv'}$ has dimension
 which exceeds the lower bound $n-sm$ or $n-s'm$
\item or when
 one of these codes is actually an alternant code for a larger degree \textit{i.e.}
 $\Alt{s}{\xv}{\yv}= \Alt{s''}{\xv}{\yv'}$ for $s'' > s$ and some multiplier
 $\yv'$
\end{enumerate}
 then 
 the right-hand term of (\ref{eq:starAlt})
 may be  smaller than the full
 space.
 This is precisely what happens for  wild Goppa codes of
 extension degree $2$ as shown by the following statement.

\begin{thm}[{\cite{COT13}}]\label{thm:properties_superGoppa}
Let $\Goppa{\xv}{\gamma^{q-1}}$ be a wild Goppa code of length $n$ defined over $\fq$ 
with support $\xv \in \Fqq^n$ where  $\gamma \in \Fqq[z]$ is irreducible of degree $r>1$. Then,
\begin{enumerate}[(i)]
\item\label{it:superGoppa} $\Goppa{\xv}{\gamma^{q-1}}  =  \Goppa{\xv}{\gamma^{q+1}} $
\item\label{it:rrmoins2} $\dim(\Goppa{\xv}{\gamma^{q+1}})  \geq  n - 2r(q+1) + r(r+2)$
\item\label{it:PsupSRS} $\Goppa{\xv}{ \gamma^{q+1}}  = \uv \star \Alt{r(q+1)}{\xv}{\onev}$ for some multiplier $\uv \in \Fq^n$.
\end{enumerate}
\end{thm}

\begin{proof}
The results (\ref{it:superGoppa}) and (\ref{it:rrmoins2}) are straightforward consequences of
Theorems 1 and 24 of \cite{COT13}.
Only (\ref{it:PsupSRS}) requires further details.
First, let us consider the case where $\xv$ is a full--support,
that is if $n=q^2$, then from
\cite[Corollary 10]{COT13}, we have
\begin{equation}\label{eq:FullWildGoppa}
\Goppa{\xv}{\gamma^{q+1}} = \av \star \left(\RS{q^2 - r(q+1)}{\xv}\cap \Fq^n\right),
\end{equation}
for some multiplier $\av \in \Fq^n$. 
Then, from Proposition~\ref{prop:dualGRS},
we have
$$
\RS{q^2-r(q+1)}{\xv} = \GRS{r(q+1)}{\xv}{\loc{\xv}'(\xv)^{-1}}^{\bot}.
$$
Since $\xv$ is assumed to be full then $\loc{\xv}'(\xv) = \onev$.
Therefore, from Definition~\ref{def:subfield_subcode}
we see that~\eqref{eq:FullWildGoppa} is equivalent to:
\begin{eqnarray}
  \label{eq:FullWildGoppaBis}
  \Goppa{\xv}{\gamma^{q+1}} & = & \uv \star
    (\RS{r(q+1)}{\xv}^{\bot} \cap \Fq^n)\\
  & = & \uv \star \Alt{r(q+1)}{\xv}{\onev},
\end{eqnarray}
which yields (\ref{it:PsupSRS}).
The general case, i.e. when $\xv$ is not full can be deduced from the full
support case by shortening thanks to Proposition~\ref{pr:shortened_alternant_code}. 
\end{proof}

\subsection{Wild Goppa codes with non generic squares}
In what follows, we prove that a
wild Goppa code over a quadratic extension $\Goppa{\xv}{\gamma^{q-1}}$
whose length belongs to some interval $[n_-, n_+]$
has a square with a  non generic behaviour.
The bounds $n_-$ and $n_+$ of the interval depend only on
the degree of $\gamma$ and will be explicitly described. The corresponding
wild Goppa codes have rather short length compared to the full support ones
and very low rate (ratio $k/n$ where $k$ denotes the dimension).

We emphasize that public keys proposed for McEliece
have not a length in the interval $[n_-, n_+]$.
However, thanks to Corollary~\ref{cor:shortened_Goppa},
a shortening of the public key is a wild Goppa code with the
same Goppa polynomial but with a shorter length and a lower rate. 
This is the point of our  "distinguisher by shortening"
described in the 
subsequent sections. 

\subsubsection{The parameters of wild Goppa codes with non generic
  squares}\label{ss:dist_sans_racc}
Let $\CC = \Goppa{\xv}{\gamma^{q-1}}$ be a wild Goppa code over a quadratic extension.
We look for a sufficient condition on the length of $\CC$ for its square
to have a non generic behavior.
From Theorem~\ref{thm:properties_superGoppa},
we have $\CC = \Goppa{\xv}{\gamma^{q+1}}$. Thus, it
is an alternant code of degree $r(q+1)$ and from
Theorem~\ref{thm:sq_alternant}:
\begin{equation}
  \label{eq:incl_alt}
  \sqc{\CC} \subseteq \Alt{2r(q+1)+1-n}{\xv}{\yv}
\end{equation}
for some multiplier $\yv \in \Fqq^n$. 
Let $\code{R}$ be a random code with the same length and dimension
as $\CC$. With 
high probability, we get
\begin{equation}\label{eq:sq_random}
\dim \sqc{\code{R}} = \min \left\{ n ~,~
{\dim \CC +1 \choose 2}\right\}\cdot
\end{equation}
Therefore, $\CC$ is distinguishable
from $\code{R}$ when the following conditions are both satisfied:
\begin{align}
\label{eq:dist1}\tag{D1} \dim \Alt{2r(q+1)+1-n}{\xv}{\yv}
       & < n \\
\label{eq:dist2}\tag{D2}  \dim
\Alt{2r(q+1)+1- n}{\xv}{\yv} & <
 {\dim \CC +1 \choose 2}.
\end{align}
Using the very definition of alternant codes
(Definition~\ref{def:subfield_subcode}), one proves easily that
an alternant code is different from its ambient space if and
only if its degree is positive. Thus,
\eqref{eq:dist1} is equivalent to:
\begin{equation}
  \label{eq:first_bound}
   n \leq 2r(q+1).
\end{equation}
Next, from Theorem~\ref{thm:properties_superGoppa}(\ref{it:rrmoins2}),
we have
\begin{equation}
  \label{eq:dim_PK}
  \dim \CC \geq n - 2r(q+1) + r(r+2).
\end{equation}
Moreover, from Theorem~\ref{prop:designed}(\ref{it:designed_dim}) on the
dimension of alternant codes, we
have
\begin{equation}
  \label{eq:dim_emb_alt}
  \dim \Alt{2r(q+1)+1-n}{\xv}{\yv} \geq 
  3n - 4r(q+1) - 2.
\end{equation}
Assume that the above lower bounds (\ref{eq:dim_PK}) and (\ref{eq:dim_emb_alt})
on the dimensions of $\CC$
and $\Alt{2r(q+1)+1-n}{\xv}{\yv}$ are their actual dimension
(which holds true in general).
In such a case, 
\eqref{eq:dist2} becomes equivalent to
\begin{equation}
  \label{eq:horrible_thing}
  {n -2r(q+1)+r(r+2)+1 \choose 2} >  3 n - 4r(q+1) - 2.
\end{equation}
Therefore, if the length $n$ of $\CC$ satisfies both (\ref{eq:first_bound})
and (\ref{eq:horrible_thing}), then its square has 
a non generic dimension. Now, consider the map
$$
\varphi : n \longmapsto   {n -2r(q+1)+r(r+2)+1 \choose 2} -  3n + 4r(q+1) + 2.
$$
Its first derivative is 
$$
\varphi'(n) = 
n -2r(q+1)+r(r+2)  +\frac{1}{2} -3 = n -2r(q+1)+r(r+2)  - \frac{5}{2}.
$$
Hence, for $n > 2r(q+1) -r(r+2)
+\frac{5}{2}$, 
the map $\varphi$ is increasing.

There are two cases to consider :\\
(i) either \eqref{eq:horrible_thing} is not satisfied for the largest possible value of $n$ 
satisfying \eqref{eq:first_bound}, namely
$n=2r(q+1)$ and then the fact that $\varphi$ is increasing implies that it can not be satisfied 
by any value of $n$;\\
(ii) or \eqref{eq:horrible_thing} is satisfied for $n=2r(q+1)$ and then 
the set of values $n$ satisfying both \eqref{eq:first_bound} and \eqref{eq:horrible_thing}
is an interval of the form $[n_-,n_+]$ where $n_+=2r(q+1)$ and 
$n_-$ is the smallest $n$ satisfying \eqref{eq:horrible_thing}.

Notice that \eqref{eq:horrible_thing} is satisfied for $n=2r(q+1)$ if and only if
\begin{equation}\label{eq:dist_condition}
{r(r+2)+1 \choose 2} > 2r(q+1) - 2.
\end{equation}
The whole discussion can be summarized by

\paragraph{\bf Conclusion}
Assuming that the above lower bounds (\ref{eq:dim_PK}) and (\ref{eq:dim_emb_alt})
on the dimensions of $\CC$
and $\Alt{2r(q+1)+1-n}{\xv}{\yv}$ are their actual dimension and
if (\ref{eq:dist_condition}) holds,
then there is a nonempty interval $[n_-, n_+]$ of integers such that
 $\sqc{\CC}$ has a non generic behaviour
for any $n$ in  $[n_-, n_+]$
Moreover, 
\begin{enumerate}
\item $n_+ = 2r(q+1)$;
\item $n_-$ is the least integer such that 
$$
{n - 2r(q+1) + r(r+2)+1 \choose 2} > 3n - 4r(q+1) -2.
$$
\end{enumerate}
For a length exceeding $n_+$, the square will probably
be equal to the whole ambient space, while for a length
less than $n_-$, the square will probably have a dimension equal
to that of a square random code.

\subsection{A distinguisher by shortening}\label{ss:distinguisher}
It can easily be checked that proposed public keys for McEliece
have a length far above the upper bound $n_+$ described in the previous
section. However the previously described interval $[n_-, n_+]$
only depends on the degree $r$ of $\gamma$. Moreover,
according to Corollary~\ref{cor:shortened_Goppa}, shortening
$\CC$ provides a shorter Goppa code with the same Goppa
polynomial. 
This leads to the first fundamental result of this article.

\begin{thm}\label{thm:distinguisher}
Let $\CC$ be the wild Goppa code $\Goppa{\xv}{\gamma^{q-1}}$,
where $\gamma \in \Fqq [z]$ has degree $r < q$.
If the following inequality holds, 
$$
{r(r+2)+1 \choose 2} > 2r(q+1) - 2,
$$  
then there is a nonempty interval $[a_-, a_+] \subseteq [1, n]$ such that
for all $\Ind \subseteq \{0, \ldots, n-1\}$ with $|\Ind| \in [a_-, a_+]$,
the dimension of $\sqc{\sh{\CC}{\Ind}}$
 is less than that of almost all
squares of random codes of the same length and dimension.
Moreover, 
\begin{enumerate}
\item $a_-  = n - 2r(q+1)$;
\item $a_+$ is the largest integer such that 
$$
  {n-a_+ -2r(q+1)+r(r+2)+1 \choose 2} >  3(n-a_+) - 4r(q+1) - 2.  
$$
\end{enumerate}
\end{thm}

\begin{proof}
  Apply the reasoning of \S~\ref{ss:dist_sans_racc} to $\sh{\CC}{\Ind}$,
  \textit{i.e.} replace everywhere $n$ by $n - |\Ind|$.
\end{proof}

\subsection{Experimental observation and example}\label{ss:experimental}

Actually, in all our experiments we observed that $\sqc{\sh{\CC}{\Ind}}$ has
always codimension $1$ in the code
$\Alt{2r(q+1)+1-n}{\xv_{\Ind}}{\yv_{\Ind}}$ (see \eqref{eq:incl_alt}).
This allows to replace the strict inequalities in 
\eqref{eq:dist1} and \eqref{eq:dist2} by large ones
and provides a slightly larger distinguisher interval $[a_-, a_+]$,
which turns out to be the actual distinguisher interval according to our
experiments. Namely
\begin{enumerate}
\item $a_- = n - 2r(q+1) - 1$;
\item $a_+$ is the largest integer such that
$$
  {n-a_+ -2r(q+1)+r(r+2)+1 \choose 2} >  3(n-a_+) - 4r(q+1) - 3.  
$$
\end{enumerate}
This interval is nonempty as soon as: 
$$
{r(r+2)+2 \choose 2} > 2r(q+1).
$$

We checked that this allows to distinguish from random codes all the wild Goppa
codes of extension degree $2$ suggested in \cite{BLP10} when $r > 3$.
For instance, 
the first entry in \cite[Table~7.1]{BLP10} is a wild Goppa code $\CC$ defined over $\F_{29}$ of length $794$,
dimension $529$ with a Goppa polynomial $\gamma^{29}$ where $\deg \gamma = 5$. 
Table~\ref{tab:distinguisher} shows that for $a$ in the range $\{493,\ldots, 506\}$
the dimensions of $\sqc{\sC{\Ind}}$
 differ from those of a random code with the same parameters.
 Note that for this example $\am=493$.

It is only when the degree of $\gamma$ is very small and the field size large
that we cannot distinguish the Goppa code in this way.
In  Table~\ref{tab:limit}, we gathered upper bounds on the field size for
which we expect to distinguish $\Goppa{\xv}{\gamma^{q-1}}$ from 
a random code in terms of the degree of $\gamma$.

\begin{table}[!h]
\caption{Dimension of $\sqc{\sC{\Ind}}$ when $\CC$ is either the
wild Goppa code in the first entry of \cite[Table~7.1]{BLP10},
or
a random code of the same length and dimension  for various values $\card{\Ind}$.
\label{tab:distinguisher}}
\begin{center}
\begin{tabular}{@{}*{13}{c}@{}}
\toprule 
$|\Ind|$ & 493 &494 &495 &496&497&498&499&500&501&502&503&504\\ \midrule
 Goppa & 300 & 297& 294& 291&288&285&282&279&276&273&270&267\\
 Random & 301 & 300&299&298&297&296&295&294&293&292&291&290 \\
 \bottomrule
\end{tabular}

\bigskip

\begin{tabular}{@{}*{11}{c}@{}}
\toprule
$|\Ind|$ & 505& 506 &507 &508 &509 & 510 &511&512 &513 &514 \\ \midrule
 Goppa & 264 & 261 &  253& 231& 210& 190& 171& 153& 136 & 120 \\
 Random & 289 & 276 & 253& 231& 210& 190& 171& 153& 136 & 120 \\
 \bottomrule
\end{tabular}
\end{center}
\end{table}

\begin{table}[h!]
\caption{Largest field size $q$ for which we can expect to distinguish 
$\Goppa{\xv}{\gamma^{q-1}}$ when $\gamma$ is an irreducible
polynomial in $\Fqq[z]$ of degree $r$. \label{tab:limit}}
\begin{center}
\begin{tabular}{@{}*{5}{c}@{}}
\toprule
$r$ & 2 & 3 & 4 & 5 \\ \midrule
$q$ & 9 & 19 & 37 & 64 \\ \bottomrule
\end{tabular}
\end{center}
\end{table}


\section{The code filtration}
\label{sec:nested}

In this section, $\CC$ denotes a wild Goppa code $\Goppa{\xv}{\gamma^{q-1}}$
over a quadratic extension.
The crucial ingredient of our attack is the computation of a family of nested codes
from the knowledge of the public key.
According to the common terminology in commutative algebra,
we call such a family a {\em filtration}. 
Roughly speaking, given the public code $\CC = \Goppa{\xv}{\gamma^{q-1}}$,
we aim at computing a filtration:  
$$
\Coi{a}{0} \supseteq \Coi{a}{1} \supseteq \cdots \supseteq \Coi{a}{s} \supseteq \cdots
$$
such that $\Coi{a}{0}$ is some puncturing of $\CC$.
Moreover, we wish the filtration to have a good behavior with respect to the
Schur product. Ideally we would expect something like:
$$
\textrm{`` } i+j = k+\ell \quad \Longrightarrow \quad \spc{\Coi{a}{i}}{\Coi{a}{j}}
 = \spc{\Coi{a}{k}}{\Coi{a}{\ell}}.\textrm{''}
$$
This is exactly what would happen
if $\CC$ is a GRS code (see \S\ref{ss:tutorial}).
Unfortunately, in the case of a wild Goppa code such a requirement is too strong
and we will only have a weaker but sufficient version asserting that
$\spc{\Coi{a}{i}}{\Coi{a}{j}}$ and $\spc{\Coi{a}{k}}{\Coi{a}{\ell}}$ are
contained in a same alternant code. This is detailed further in Corollary~\ref{cor:prod_Coi}.

Roughly speaking, the code $\Coi{a}{j}$
(see Definition~\ref{def:defCoi} below) consists in
the codewords of $\CC$ obtained from polynomials having a zero of order
at least $j$ at position $a$.
The key point is that this filtration reveals a lot about the
algebraic structure of $\CC$.
In particular, we will be able to recover the support from it. 
To understand the rationale behind such a filtration its computation and its
use for cryptanalysis,
let us start with an illustrative example on generalized Reed Solomon codes.

\subsection{Illustrative example with GRS codes}\label{ss:tutorial}
Let $\xv, \yv$ be a support and a multiplier in $\Fqq^n$. Let $k < n/2$.
Assume that the codes $\GRS{k}{\xv}{\yv}$ and
$\GRS{k-1}{\xv}{\yv}$ are known.
We claim that from the single knowledge of these
two codes, it possible to compute the whole filtration
\begin{equation}\label{eq:GRS_filtration}
\GRS{k}{\xv}{\yv} \supseteq \GRS{k-1}{\xv}{\yv}
 \supseteq \cdots \supseteq \GRS{i}{\xv}{\yv} \supseteq
\cdots \supseteq \GRS{1}{\xv}{\yv} \supseteq \{0\},
\end{equation}

\begin{rem}
In terms of polynomials, this filtration corresponds to:
\begin{equation}\label{eq:degree_filtration}
\Fqq[z]_{<k} \supseteq \Fqq [z]_{<k-1} \supseteq \cdots \supseteq \Fqq[z]_{<i} \supseteq
\cdots \supseteq \Fqq[z]_{<1} \supseteq \{0\}.
\end{equation}
\end{rem}

Let us explain how we could compute $\GRS{k-2}{\xv}{\yv}$.
From Proposition~\ref{prop:GRSsquare}, we obtain
\begin{equation}\label{eq:supertrick}
\GRS{k-2}{\xv}{\yv} \star \GRS{k}{\xv}{\yv} = \sqc{\GRS{k-1}{\xv}{\yv}}
\end{equation}
and from this equality, one can prove that:
  \begin{equation}\label{eq:comp_filtration}
  \GRS{k-2}{\xv}{\yv} = \left\{ \cv \in \GRS{k-1}{\xv}{\yv} ~|~ \cv \star \GRS{k}{\xv}{\yv}
   \subseteq \sqc{\GRS{k-1}{\xv}{\yv}}\right\}.
  \end{equation}
Indeed, inclusion ``$\subseteq$'' is a direct consequence of (\ref{eq:supertrick}).
The converse inclusion can be obtained by studying the degrees
in the associated spaces
of polynomials (see \cite[\S6]{CGGOT14}).
Thus, Equation~(\ref{eq:comp_filtration}) shows
that $\GRS{k-2}{\xv}{\yv}$ can be computed from the single knowledge of $\GRS{k}{\xv}{\yv}$
and $\GRS{k-1}{\xv}{\yv}$. By iterating this process, one can compute
all the terms of the filtration (\ref{eq:GRS_filtration}). Finally, since the
last nonzero term $\GRS{1}{\xv}{\yv}$ is obtained by evaluation of constant polynomials,
this space has dimension $1$ and is spanned by $\yv$. This yields $\yv$ up to a multiplication
by a scalar.

\medbreak

This is an illustration of how the computation of a filtration can provide crucial
information on a code.
On the other hand,  there is no reason
to know $\GRS{k-1}{\xv}{\yv}$  especially in a cryptographic situation.
However, some very particular subcodes of codimension $1$ can be easily computed
from the knowledge of $\GRS{k}{\xv}{\yv}$.
Namely, shortening $\sh{\GRS{k}{\xv}{\yv}}{i}$ at a single position
 $i \in \{0, \ldots, n-1\}$ can be computed by Gaussian elimination.
This code corresponds to the space
of polynomials vanishing at $x_i$, that is the space $(z-x_i)\Fq[z]_{<k-1}$,
or, after a suitable change of variables, the space $z\Fq[z]_{<k-1}$ of polynomials
vanishing at $0$. Therefore, 
using the method described above, from $\GRS{k}{\xv}{\yv}$ and its shortening at
the $i$--th position, one can compute the filtration of codes corresponding
to the spaces of polynomials:
\begin{equation}\label{eq:filt_val}
\Fq[z]_{<k} \supseteq z\Fq[z]_{<k-1} \supseteq \cdots \supseteq
z^j \Fq[z]_{<k-j}\supseteq \cdots \supseteq z^{k-1} \Fq[z]_{<1} \supseteq \{0\}.
\end{equation}
The computation of such filtrations permits a complete recovery of
$\xv, \yv$ (see \cite{CGGOT14} for further details).
This is exactly the spirit of our attack on wild Goppa codes.

\begin{rem}
From an algebraic geometric point of view, 
the filtrations (\ref{eq:degree_filtration}) and
(\ref{eq:filt_val}) are very close to each other.
Filtration (\ref{eq:filt_val}) is the filtration associated to the 
valuation at $0$ while filtration (\ref{eq:degree_filtration})
is associated to the degree which can be regarded as a valuation 
at infinity. Thus,  investigating
a filtration like (\ref{eq:filt_val}) is extremely natural.
\end{rem}

\subsection{The computation of particular subcodes}
In the previous examplel and also in what follows, the computation
of a term of a filtration can be done from the previous ones
by solving a problem of the form:
\begin{problem}\label{pr:step1}
Given $\code{A}$, $\code{B}$, and $\code{D}$ be three codes in $\Fq^n$
find the subcode $\code{S}$ of elements $\sv$ in $\code{D}$   satisfying: 
\begin{equation}
\spc{\sv}{\code{A}}  \subseteq  \code{B} \label{eq: AstarX}
\end{equation}
\end{problem}

Such a code can be computed by linear algebra
or equivalently by computing dual codes and Schur products.
Namely, we have:
\begin{proposition}\label{prop:SolSpace}
  The solution space $\code{S}$ of Problem \ref{pr:step1} is:
  $$
    \code{S} = {\left(\spc{\code{A}}{\code{B}^{\bot}} \right)}^{\bot} \cap \DC.
  $$
\end{proposition}
 
\begin{proof}
  Let $\sv \in \code{S}$ then clearly $\sv \in \DC$.
  Let $\av \in \code{A}$ and $\bv^{\bot} \in \code{B}^{\bot}$.
  Then,
$$
\langle \sv, \spc{\av}{\bv^{\bot}} \rangle 
  = \sum_{i=0}^{n-1}  z_i a_i b^{\bot}_i
  = \langle \spc{\sv}{\av}, \bv^{\bot} \rangle
$$
and this last term is zero by definition of $\code{S}$.
This proves
$\code{S} \subseteq {\left(\spc{\code{A}}{\code{B}^{\bot}} \right)}^{\bot}
\cap \DC$. The converse inclusion is proved in the very same way.
\end{proof}

\subsection{The filtration of alternant codes $\Coi{a}{j}$}

In the same spirit as the example of \S\ref{ss:tutorial}, we will
compute the terms of a filtration by solving iteratively problems
of the form of Problem~\ref{pr:step1}. The filtration we will compute
is in some sense related to the polynomial spaces filtration:
$$
\Fqq[z]_{<k} \supseteq z\Fqq[z]_{<k-1} \supseteq \cdots \supseteq
z^j \Fqq[z]_{<k-j}\supseteq \cdots \supseteq z^{k-1} \Fqq[z]_{<1} \supseteq \{0\}.
$$
For that purpose, we introduce the following definition.

\begin{definition}\label{def:defCoi}
   For all $a \in \{0, \ldots , n-1\}$
and for all $s\in \Z$, 
we define the code $\Coi{a}{s}$ as:
$$
\Coi{a}{s} \eqdef  
\cset{\frac{\gamma^{q+1}(x_i)}{\loc{\xv}'(x_i)} (x_i- x_a)^s f(x_i)}
{i \in \{0,\dots,n-1\} \setminus \{a\}}{f \in \Fqq[z]_{<n- r(q+1)-s}} 
\cap \Fq^{n-1}.
$$
Roughly speaking, for $s>0$, the code $\Coi{a}{s}$ is the subcode of $\sh{\CC}{a}$ obtained from
rational fractions vanishing at $x_a$ with order at least $s$.
\end{definition}

The link with $\CC$ becomes  clearer if we use
Theorem~\ref{thm:properties_superGoppa}, which asserts that
$\CC = \Goppa{\xv}{\gamma^{q+1}}.$
Thanks to Lemma~\ref{lem:descr_Goppa_as_eval} on the description of
Goppa codes as evaluation codes, we have:
\begin{equation}\label{eq:CC}
\CC = \cset{\frac{\gamma^{q+1}(x_i)}{\loc{\xv}'(x_i)}f(x_i)}{0 \leq i < n}{f \in \Fqq[z]_{<n- r(q+1)}}
\cap \Fq^n.
\end{equation}
From now on,
we focus on the case $a=0$ and assume that
\begin{assumption}
\label{ass:fundamental}
\begin{enumerate}[(i)]
\item\label{it:one}
$\sh{\CC}{0} \neq \CC$
\item $x_0=0$, $x_1=1$.
\end{enumerate}
\end{assumption}

{\em Discussion about these assumptions.}
If $\CC$ is not the zero code, after possibly reordering the support we can always assume
that the first position is not always equal to $0$ in every codeword of $\CC$ and therefore
$\sh{\CC}{0} \neq \CC$.
The second assumption can always be made, and this without reordering the support- this 
follows directly from Lemma~\ref{lem:2transitive}.

Every statement in what follows could be reformulated
for a general position $a$. This would however provide
heavier notation which we have tried to avoid.

The following statement summarizes the properties of 
this filtration which are used in the attack.
Since its proof is rather technical, we chose to postpone
it in appendix.

\begin{proposition} Under Assumption \ref{ass:fundamental} (i), we have
  \label{prop:properties_Coi}
\begin{enumerate}[(i)]
  \item\label{it:short_Coi} $\Coi{0}{1} = \sh{\CC}{0}$;
  \item\label{it:simplify} $\Coi{0}{0} = \pu{\CC}{0}$;
  \item\label{it:codim_Coi} $\forall s\in \Z,\ \dim \Coi{0}{s}
    -\dim \Coi{0}{s+1} \leq 2$;
  \item\label{it:stagn} $\Coi{0}{q-r} = \Coi{0}{q+1}$;
  \item\label{it:desc_as_alt} $\forall s\in \Z,\ \Coi{0}{s} = 
    \Alt{r(q+1)+s-1}{\xvzero}{\yvzero}$ for 
    $$
    \yvzero \eqdef \gamma^{-(q+1)}(\xvzero) \star \xvzero^{-(s-1)},
    $$
\end{enumerate}
where we recall that $\xv_0$ denotes the vector
$\xv$ punctured at position $0$ and that $r$ denotes the
degree of $\gamma$.
\end{proposition}

\begin{proof}
  Appendix~\ref{sec:stagn}.
\end{proof}

\begin{cor}
  \label{cor:dim_Coi}
  For all $ s > 0$, we have
  $$
  \dim \Coi{0}{s} \geq n-1 - 2r(q+1) -2(s-1) + r(r+2).
  $$
\end{cor}

\begin{proof}
  The case $s=1$ is a direct consequence of
  Proposition~\ref{prop:properties_Coi}(\ref{it:short_Coi})
  since shortening at one position reduces the dimension
  from at most $1$. Then the result is proved by induction
  on $s$ using Proposition~\ref{prop:properties_Coi}(\ref{it:codim_Coi}).
\end{proof}
 
\begin{cor}\label{cor:prod_Coi}
  For all pair $s,s'$ of integers, 
  $$
  \spc{\Coi{0}{s}}{\Coi{0}{s'}} \subseteq \Alt{2r(q+1) + s + s' - n}{\xvzero}{\yvzero}
  $$
  where $\yvzero = \gamma^{-2(q+1)}(\xvzero) \star \xvzero^{-(s+s'-2)} \star \loc{\xvzero}'(\xvzero)$.
\end{cor}

\begin{proof}
  Apply Proposition~\ref{prop:properties_Coi}(\ref{it:desc_as_alt}) and
  Theorem~\ref{thm:sq_alternant} using the fact that $\Coi{0}{s}$ and $\Coi{0}{s'}$ are of length $n-1$.
\end{proof}

\subsection{The distinguisher intervals}\label{ss:filtr_dist}
The filtration ${\big(\Coi{0}{s} \big)}_{s\in \Z}$ is strongly related
to $\CC$ since as explained in
Proposition~\ref{prop:properties_Coi}(\ref{it:short_Coi})
and (\ref{it:simplify}), two elements of the filtration 
can easily be computed from the pulic key $\CC$. Namely, the codes
$\Coi{0}{0}$ and $\Coi{0}{1}$ are respectively obtained
by puncturing and shortening $\CC$ at position $0$.
The subsequent elements of the filtration will
be computed iteratively by solving problems of the form
of Problem~\ref{pr:step1} in the very same manner as in the 
example given in Section~\ref{ss:tutorial}.
For instance, we will compute $\Coi{0}{2}$ from the ``equation''
$$
\textrm{``}\spc{\Coi{0}{0}}{\Coi{0}{2}} \subseteq \sqc{\Coi{0}{1}}
\textrm{''}
$$
and more generally $\Coi{0}{t}$, will be computed from
\begin{equation}
  \label{eq:fausse}
\textrm{``}\spc{\Coi{0}{0}}{\Coi{0}{t}} \subseteq
 \spc{\Coi{0}{\lfloor
 t/2\rfloor}}{\Coi{0}{\lceil
 t/2 \rceil}}
\textrm{''}  
\end{equation}

Unfortunately, this relation is not strictly correct: according to
Corollary~\ref{cor:prod_Coi}, the right-hand term should be replaced by
$\Alt{2r(q+1)+t-n}{\xvzero}{\yvzero}$ which is unknown.
Moreover, as explained in \S\ref{sec:distinguisher} the above Schur
products fill in their ambient space. However, for some particular lengths
it is possible to compute $\Coi{0}{t}$ by solving a problem
of the form of Problem~\ref{pr:step1}. These lengths are those such that: 
\begin{enumerate}
\item The alternant code $\Alt{2r(q+1)+t-n}{\xvzero}{\yvzero}$  does not
fill in the ambient space.
\item\label{it:fill_in} The Schur product  $\spc{\Coi{0}{\lfloor t/2\rfloor}}{\Coi{0}{\lceil
t/2 \rceil}}$ should fill in $\Alt{2r(q+1)+t-n}{\xvzero}{\yvzero}$
or at least be a sufficiently large subcode of it.
\end{enumerate}
Let $\code{R}$ and $\code{R}'$ be two random codes such that
$\code{R}' \subseteq \code{R}$ and whose dimensions equal those of
 $\Coi{0}{\lfloor t/2 \rfloor}$ and $\Coi{0}{\lceil t/2 \rceil}$. 
For (\ref{it:fill_in}) to be satisfied, we expect that the dimension
of $\spc{\code{R}}{\code{R}'}$ exceeds that of
$\Alt{2r(q+1)+t-n}{\xvzero}{\yvzero}$. Thus, the computation of $\Coi{0}{t}$
is possible if the length of the codes is in some particular interval.
Therefore, it is possible to compute $\sh{\Coi{0}{t}}{\Ind}$ for a suitable
set $\Ind$ and $|\Ind|$ is in some interval which is nothing but a distinguisher
interval as computed in \S\ref{sec:distinguisher}. We will compute these 
intervals  in order to obtain $\Coi{0}{t}$ by considering separately
the cases of even and odd $t$.

\subsubsection{The symmetric case}\label{sss:symmetric}
Assume that $t$ is even:
$$t=2s$$
for some positive integer $s$.
From Corollary~\ref{cor:dim_Coi}, we have 
$$
\dim \Coi{0}{s} \geq (n-1) - 2(r(q+1) +s-1) +r(r+2).
$$
Since, from Proposition~\ref{prop:properties_Coi}(\ref{it:desc_as_alt}),
this code is alternant of degree $r(q+1) + s-1$. Then from
Corollary~\ref{cor:prod_Coi}:
\begin{equation}\label{eq:incl_Sq_Coi}
\sqc{\Coi{0}{s}} \subseteq \Alt{2(r(q+1)+s)-n}{\xvzero}{\yvzero}
\end{equation}
Thus, we are in the very same situation as in \S\ref{ss:distinguisher}
and the distinguisher interval for $\Coi{0}{s}$ can be deduced from that
of $\CC$ by applying the changes of variables
$$
\begin{array}{lcl}
  n & \longmapsto & n-1 \\
  r(q+1) & \longmapsto & r(q+1)+s-1.
\end{array}
$$

\paragraph{\bf Conclusion} If
$$
{r(r+2)+2 \choose 2} > 2r(q+1) + t - 2
$$
then there is a nonempty interval $[b_-, b_+]$
such that for all $\Ind \subseteq \{1, \ldots, n-1\}$
with $|\Ind| \subseteq [b_-, b_+]$, the square of $\sh{\Coi{0}{s}}{\Ind}$
has a non generic behaviour. Moreover, 
\begin{enumerate}[(1)]
\item $b_- = n-2r(q+1)-t$;
\item $b_+$ is the largest integer such that
$$
{n-b_+ - 2r(q+1)  - t + 2 +r(r+2) \choose 2}
> 3(n-1-b_+) - 4r(q+1) - 2t +1.
$$
\end{enumerate}

\begin{rem}\label{rem:experimental}
Actually, the above distinguisher interval, relies on
an experimental observation similar to that of 
\S\ref{ss:experimental}.
Namely, we observed experimentally that (\ref{eq:incl_Sq_Coi}) is a strict
inclusion with codimension $1$ as soon as the degree
of the alternant code in the right hand term is non-negative.
\end{rem}

\subsubsection{The asymmetric case}
As in \S\ref{ss:dist_sans_racc}, we start by computing
the interval for which the Schur product
$\spc{\Coi{0}{s}}{\Coi{0}{s+1}}$
has a non generic behaviour, then we can reduce to that case by shortening.

\medskip

In the spirit of the distinguisher interval computed in
\S\ref{sec:distinguisher}.
Instead of Equation~(\ref{eq:incl_alt}), Corollary~\ref{cor:prod_Coi}
yields
\begin{equation}
  \label{eq:incl_alt_bis}\tag{\ref{eq:incl_alt}'}
  \spc{\Coi{0}{s}}{\Coi{0}{s+1}}
  \subseteq \Alt{2r(q+1) + t - n}{\xvzero}{\yvzero}.
\end{equation}
This leads to new distinguisher conditions:
\begin{align}
 \label{eq:dist1_bis}\tag{D1'} \dim \Alt{2r(q+1)+t-n}{\xvzero}{\yv'}
        & < n -1 \\
\label{eq:dist_2bis}\tag{D2'}  \dim
\Alt{2r(q+1)+t-n}{\xvzero}{\yv'} & < \dim \Coi{0}{s} \dim \Coi{0}{s+1} -
{\dim \Coi{0}{s+1} \choose 2}.
\end{align}
According to Proposition~\ref{pr:typ_dim}(\ref{eq:asymprod}), the right hand term
of (\ref{eq:dist_2bis}) is the typical dimension of the Schur Product
of two random codes of the same dimension as $\Coi{0}{s}$
and $\Coi{0}{s+1}$ if this product does not fill in the ambient space.
From Corollary~\ref{cor:dim_Coi}, we have

\begin{align*}
  \dim \Coi{0}{s} &\geq (n -1) - 2r(q+1) -2s+2 +r(r+2)\\
  \dim \Coi{0}{s+1} &\geq (n -1) - 2r(q+1) -2s +r(r+2)\\
\end{align*}
Assuming that the above lower bounds are reached, which holds true in general,
a computation from the formula of  Proposition~\ref{pr:typ_dim}(\ref{eq:asymprod})
gives
$$
\dim \Coi{0}{s} \dim \Coi{0}{s+1} -
{\dim \Coi{0}{s+1} \choose 2} = \frac{1}{2} d(d+5)
$$
where
$$
\ d\eqdef (n -1) - 2r(q+1) -2s +r(r+2),
$$

\paragraph{\bf Conclusion}
Proceeding as in \S\ref{ss:distinguisher}
and thanks to an experimental observation similar
to Remark~\ref{rem:experimental}, we obtain that if
$$
\frac{1}{2} r(r+2)\left(r(r+2)+5\right) > 2r(q+1)+t-2
$$
then there exists an interval $[b_-, b_+]$
such that for $\Ind$ such that $|\Ind| \subseteq [b_-, b_+]$,
the Schur product $\spc{\sh{\Coi{0}{s}}{\Ind}}{\sh{\Coi{0}{s+1}}{\Ind}}$
has a non generic behavior.
Moreover,
\begin{enumerate}
\item $b_- = n- 2r(q+1) -t$;
\item $b_+$ is the largest integer such that
$$
\frac{1}{2} d(d+5) > 3(n-1-b_+) - 4r(q+1) -2t +1,
$$
where 
$$
d = (n-1-b_+) - 2r(q+1) -2s +r(r+2).
$$
\end{enumerate}

\begin{rem}
From now on, in both situations ($t$ even or odd), the corresponding
interval will be referred to as the {\em distinguisher interval} for $\Coi{0}{t}$.
\end{rem}
 
\subsection{A theoretical result on the multiplicative
structure of the filtration}\label{ss:fundamental}

As explained previously, (\ref{eq:fausse}) does not hold in general even for the $\Coi{0}{j}$'s even for shortenings
at a set $\Ind$ such that $|\Ind|$ belongs to the distinguisher
interval. 
However, we have the following Theorem.
We explain in the sequel (see \S\ref{ss:algo}) how to apply it practically. 
To avoid a huge amount of notation in its proof, we state it 
under a condition on the length of the $\Coi{0}{j}$'s.
It can then be applied in the general case to suitable
shortenings of these codes.
\begin{theorem}
\label{thm:fundamental}
Let $t>1$ be an integer and assume that $n < 2r(q+1) + t $.
Then,
$$
\Coi{0}{t} = \Big\{\cv \in \Coi{0}{t-1} ~|~
\spc{\cv}{\Coi{0}{0}} \subseteq \Alt{2r(q+1)+t-n}{\xvzero}{\yvzero}\Big\}
$$
where $
\yvzero  \eqdef \gamma^{-(2q+2)}(\xvzero) \star
\loc{\xvzero}'(\xvzero)^{-1} \star \xvzero^{-(t-2)}. 
$
\end{theorem}

\begin{proof}
Inclusion $\subseteq$ is a consequence of Corollary~\ref{cor:prod_Coi}.
Conversely, let $\cv \in \Coi{0}{t-1}$ be such that
\begin{equation}\label{eq:assumption_step1}
\cv \star \Coi{0}{0} \subseteq \Alt{2r(q+1)+t-n}{\xvzero}{\yvzero}.
\end{equation} Choose also an element 
$\cv' \in \Coi{0}{0} \setminus \Coi{0}{1}$.
From the definition of the $\Coi{0}{j}$'s (Definition~\ref{def:defCoi}),
these two codewords are of the form:
\begin{align*}
\cv & =  \gamma(\xvzero)^{q+1} \star \xvzero^{t-1} \star
\loc{\xv}'(\xvzero)^{-1} \star f(\xvzero) \\
\cv' & =  \gamma(\xvzero)^{q+1} \star
\loc{\xv}'(\xvzero)^{-1} \star g(\xvzero),  
\end{align*}
where
\begin{equation}\label{eq:degf}
\deg(f)  \leq n- r(q+1)-t \quad {\rm and} \quad
\deg(g)  \leq n- r(q+1)-1 
\end{equation}
whereas 
and $g$ does not vanish at $0$.
From Lemma~\ref{lem:locator_short} in Appendix~\ref{sec:stagn},
we have $\loc{\xv}'(\xvzero) =
\xvzero \star \loc{\xvzero}'(\xvzero)$ and hence,
\begin{align}
\spc{\cv}{\cv'} &= \gamma(\xvzero)^{2(q+1)} \star \xvzero^{t-3} \star
\loc{\xvzero}'(\xvzero)^{-2} \star f(\xvzero) \star g(\xvzero)\\  
&= \yvzero^{-1} \star  \loc{\xvzero}'(\xvzero)^{-1} \star \xvzero^{-1} \star 
f(\xvzero) \star g(\xvzero). \label{eq:cc'}
\end{align}
By (\ref{eq:assumption_step1}), 
$\spc{\cv}{\cv'}$ belongs to $\Alt{2r(q+1)+t-n}{\xvzero}{\yvzero}$.
Using the description of alternant codes as evaluation codes
given in Lemma~\ref{lem:decription_as_eval}, it can therefore be written as
\begin{equation}
  \label{eq:alter_c}
\spc{\cv}{\cv'} =
  \spc{\yvzero^{-1}\star \loc{\xvzero}'(\xvzero)^{-1}}{h(\xvzero)},  
\end{equation}
where
\begin{equation}\label{eq:degh}
\deg(h)  < n-1 -2r(q+1)-t+n  = 2n - 1 - 2r(q+1) -t. 
\end{equation}
Putting (\ref{eq:cc'}) and (\ref{eq:alter_c}) together, we get the vector 
equality
$
\xvzero^{-1}\star\spc{f(\xvzero)}{g(\xvzero)} = {h(\xvzero)}.
$
Or, equivalently
$$
\forall i \in \{1, \ldots, n-1\},\ f(x_i)g(x_i) = x_i h(x_i).
$$
From (\ref{eq:degf}) and (\ref{eq:degh}), the polynomial
$f(z)g(z) - zh(z)$ has degree at most $2n-2r(q+1)-t-1$.
Moreover, it has $n-1$ roots
and, we assumed that $n < 2r(q+1)+t$ which entails $2n-2r(q+1) -t -1 < n-1$.
Therefore, $f(z)g(z)-zh(z)$ has more roots than its degree, which proves
the equality $f(z)g(z) = zh(z)$.
Since by assumption, $g$ does not vanish at zero,
then $z$ divides $f$, which entails that $\cv \in \Coi{0}{t}$. 
\end{proof}

\paragraph{\bf How to use this theorem?}
Theorem~\ref{thm:fundamental} cannot be used directly in the cryptographic
context for two reasons.
\begin{enumerate}
\item\label{it:ineq_pour_short}
In general the inequality $n < 2r(q+1)+t$ is not satisfied by $\CC$.
\item\label{it:incl_prod_alt}
The alternant code $\Alt{2r(q+1)+t-n}{\xvzero}{\yvzero}$ is unknown.
\end{enumerate}

Issue (\ref{it:ineq_pour_short}) is addressed by choosing suitable
shortenings of the code.
To address issue (\ref{it:incl_prod_alt}), despite
$\Alt{2r(q+1)+t-n}{\xvzero}{\yvzero}$ is unknown, we know a possibly
large subcode of it, namely
$$
\spc{\Coi{0}{\lfloor t/2\rfloor}}{\Coi{0}{\lceil t/2 \rceil}}.
$$
Therefore, to use this result, one shortens the codes 
$\Coi{0}{\lfloor t/2\rfloor}$ and $\Coi{0}{\lceil t/2 \rceil}$ at some
set $\Ind \subseteq \{1, \ldots , n\}$ such that $|\Ind|$ lies in the
distinguisher interval for computing $\Coi{0}{t}$ defined in
\S\ref{ss:filtr_dist}.
In this context, one can compute the subcode of $\sh{\Coi{0}{t}}{\Ind}$
defined as:
$$
\Big \{\cv \in \sh{\Coi{0}{t-1}}{\Ind}  ~|~ \cv \star \sh{\Coi{0}{0}}{\Ind} \subseteq 
\spc{\sh{\Coi{0}{\lfloor t/2\rfloor}}{\Ind}}
{\sh{\Coi{0}{\lceil t/2 \rceil}}{\Ind}} \Big \}.
$$
In all our experiments, this subcode turned out to be the whole
$\sh{\Coi{0}{t}}{\Ind}$.

\subsection{The algorithm}\label{ss:algo}
One expects to find 
$\sh{\Coi{0}{t}}{\Ind}$ by solving Problem~\ref{pr:step1}.
This allows to find several of these $\sh{\Coi{0}{t}}{\Ind}$'s
associated to different subsets of indexes $\Ind$.
It is straightforward to use such sets in order to recover $\Coi{0}{t}$. 
Indeed,  we clearly expect that
\begin{equation}\label{eq:reunion}
\sh{\Coi{0}{t}}{\Ind \cap \Jind} =
\sh{\Coi{0}{t}}{\Ind} + \sh{\Coi{0}{t}}{\Jind}
\end{equation}
where with an abuse of notation we mean by $\sh{\Coi{0}{t}}{\Ind}$ and
$\sh{\Coi{0}{t}}{\Jind}$ the codes $\sh{\Coi{0}{t}}{\Ind}$ and
$\sh{\Coi{0}{t}}{\Jind}$
 whose set of positions have been completed such as to also contain the positions belonging to 
 $\Ind \setminus \Jind$ and $\Jind \setminus \Ind$ respectively and which are set to $0$.
Such an equality does not always hold of course, but apart from rather pathological 
cases it typically holds
when $\dim\left( \sh{\Coi{0}{t}}{\Ind}\right) + \dim\left( 
\sh{\Coi{0}{t}}{\Ind} \right)  \geq
\dim\left( \sh{\Coi{0}{t}}{\Ind \cap\Jind}\right)$.

\medskip

These considerations suggest Algorithm \ref{al:sCoi0v} for computing the codes $\Coi{0}{s}$ for any $s>1$. 
The value of $k(t)$  computed in line \ref{value:k} for some $t> 2$ can also be obtained ``offline'' by computing the true dimension
of a $\Coi{0}{t}$ for an arbitrary choice of $\gamma$ and $\xv$. Algorithm \ref{al:sCoi0v} uses the knowledge of $\Coi{0}{0}$ and $\Coi{0}{1}$ (see Proposition~\ref{prop:properties_Coi}). Observe that in line~\ref{dsg_interval}, the cardinality of $\Ind$ has to lie in the distinguisher interval as explained in Section~\ref{ss:filtr_dist}. The instruction in line \ref{addCoi} should be understood as 
the addition of two codes having the ``same'' length where by abuse of notation, $\sh{\Coi{0}{t}}{\Ind}$ means the code $\sh{\Coi{0}{t}}{\Ind}$ to which $0$'s have been added in the positions belonging to $\Ind$.

\begin{algorithm}[!h]
\caption{\label{al:sCoi0v} Computation of $\Coi{0}{s}$ with $s>1$}
\begin{algorithmic}[1]
\FOR{{$t=2$} to {$s$} } 
\STATE{$\Coi{0}{t} \leftarrow \{0\}$}
\STATE{$k(t) \leftarrow (n-1)-2r(q+1)-2t+ 2 + r(r + 2)$} \label{value:k}

\WHILE{$\dim \Coi{0}{t} \neq k(t)$}
\STATE{$\Ind \leftarrow$ random subset of $\{1,\ldots,n-1\}$ 
such that $|\Ind| \in [b_-, b_+]$} 
\COMMENT{Section~\ref{ss:filtr_dist}} \label{dsg_interval}
\STATE{$\code{A} \leftarrow \sh{\Coi{0}{0}}{\Ind}$}
\STATE{$\code{B} \leftarrow 
\spc{\sh{\Coi{0}{\left\lfloor \frac{t}{2} \right\rfloor}}{\Ind}}
{\sh{\Coi{0}{\left\lceil \frac{t}{2} \right\rceil}}{\Ind}}$}
\STATE{$\code{D} \leftarrow \sh{\Coi{0}{t-1}}{\Ind}$}
\STATE{$ \sh{\Coi{0}{t}}{\Ind} \leftarrow \code{D}\cap {\left(\spc{\code{A}}{ \code{B}^{\perp}}\right)}^{\perp}$}
\COMMENT{Solving of Problem \ref{pr:step1}}
\STATE{$\Coi{0}{t} \leftarrow \Coi{0}{t} + \sh{\Coi{0}{t}}{\Ind}$} \label{addCoi}
\ENDWHILE
 \ENDFOR
 \RETURN{$\Coi{0}{s}$}
\end{algorithmic}
\end{algorithm}



\section{An efficient Attack Using the Distinguisher}
\label{sec:attack}

In this section, we sketch the complete attack we implemented.
We chose to provide only a short description and
to explain it in greater detail in Appendix~\ref{sec:in_depth}. 
We emphasize that the crucial aspects
of the attack are the distinguisher and the computation of the
filtration which are presented in \S\ref{sec:distinguisher}
and \S\ref{sec:nested}. 
As soon as some terms of the filtration
are computed, it is possible to derive some interesting
information on the secret key.
The attack we present here is one manner to recover the
secret key and
we insist on the fact that there might exist many other ways
to recover it from the knowledge of some terms of the filtration.
This is further discussed in \S\ref{sec:ifcompjust}.

Before describing the attack, we state two 
key statements.

\subsection{Key tools}\label{ss:attack:key_tools}
The first statement is a very particular property of the
space $\Coi{0}{q+1}$. The fact that the $q+1$--th term
has very particular properties is not surprising. Indeed,
recall that in $\Fqq$ the map $x\mapsto x^{q+1}$ is the norm
over $\Fq$. In particular its sends $\Fqq$ onto $\Fq$.

\begin{proposition}
  \label{prop:doubleinj} We have:
  $$
 \xvzero^{-(q+1)} \star \Coi{0}{q+1} \subseteq \Coi{0}{0}.
 $$
 \end{proposition}

\begin{proof}
Since $(q+1)$--th powers in $\Fqq$ are norms over $\Fq$,
we have $\xvzero^{q+1} \in \Fq^{n-1}$. Therefore,
$\xvzero^{q+1}$ can get out of the subfield subcode
$$
\begin{aligned}
 \Coi{0}{q+1}&=
\cset{\frac{x_i^{q+1} \gamma^{q+1}(x_i)f(x_i)}{\loc{\xv}'(x_i)}}{1\leq i <n}{f\in \Fqq[z]_{<n-(r+1)(q+1)}}\cap \Fq^{n-1}\\
&= \spc{\xvzero^{q+1}}{\left(\cset{\frac{
\gamma^{q+1}(x_i)f(x_i)}{\loc{\xv}'(x_i)}}{1\leq i <n}
{f\in \Fqq[z]_{<n-(r+1)(q+1)}}\cap \Fq^{n-1} \right)}\\
&\subseteq \spc{\xvzero^{q+1}}{\left(\cset{\frac{\gamma^{q+1}(x_i)f(x_i)}
{\loc{\xv}'(x_i)}}{1\leq i <n}{f\in \Fqq[z]_{<n-r(q+1)}}
\cap \Fq^{n-1} \right)}\\
&\subseteq \spc{\xvzero^{q+1}}{\Coi{0}{0}}.
\end{aligned}
$$

\end{proof}

The second statement asserts that the minimal
polynomial over $\Fq$ of an element $t\in \Fqq$
can be deduced from the single knowledge of the
norms of $t$ and $t-1$.

\begin{lemma}\label{lem:MiniPol}
Let $t$ be an element of $\Fqq$ and 
$$
P_t(z)\eqdef z^2- (\nr(t)-\nr(t-1)-1)z+\nr(t) \in \Fq[z].
$$
Then, either $t\in \Fqq\setminus \Fq$ and $P_t$ is irreducible
in which case is the minimal polynomial of $t$ over $\Fq$, or $t\in \Fq$ and $P_t(z) = (z-t)^2$. 
\end{lemma}

\begin{proof}
First, notice that
  \begin{eqnarray*}
     \nr(t-1) & = &     (t-1)^{q+1}  =   (t-1)(t-1)^q
                                =   (t-1)(t^q-1)
                                =   t^{q+1} - t^q- t +1\\
                               & = & \nr_{\Fqq/\Fq}(t)-\tr_{\Fqq/\Fq}(t)+1.  
  \end{eqnarray*}
Therefore, $P_t (z) = z^2-\tr (t)z+\nr(t)$, which is known to
be the minimal polynomial of $t$ whenever $t\in \Fqq\setminus \Fq$.
On the other hand, when $t\in \Fq$, then $P_t(z)=z^2-2tz+t^2$
which factorizes as $(z-t)^2$.
\end{proof}

\subsection{Description of the attack}
By the $2$--transitivity of the affine group,
one can assume without loss of generality that $x_0 = 0$
and $x_1 = 1$ (see Appendix~\ref{ss:x0x1} for further details).

Let us first assume that every element of $\Fqq$
is an entry of $\xv$. The general case: $n < q^2$, is discussed
subsequently.

\begin{itemize}
\item {\bf Step 1.} Compute $\Coi{0}{q+1}$ using the method
described in \S\ref{sec:nested}. Notice that, thanks to 
Proposition~\ref{prop:properties_Coi}(\ref{it:stagn}) it is sufficient
to compute $\Coi{0}{q-r}$.
\item {\bf Step 2.} Compute the set of solutions $\cv \in \Fq^{n-1}$
of the problem
  \begin{equation}\label{eq:compNx}
    \left\{
      \begin{array}{l}
        \cv \star \Coi{0}{q+1} \subseteq \pu{\CC}{0} \\ \\
        \forall i \geq 1,\ c_i \neq 0,\quad {\rm (i.e.}\ \cv \ {\rm has\ full\ weight)} \\ \\
        c_1  = 1.
      \end{array}
    \right.
  \end{equation}
  From Proposition \ref{prop:doubleinj}, 
  $\xvzero^{-(q+1)}$ is one of the solutions of this problem.
  Indeed, it clearly satisfies the first equation, has full weight ($0$
  has been removed) and its first entry is $1$ since we assumed that $x_1 = 1$.
  One proves in Appendix~\ref{ss:furtherStep2},
  that the space of words $\cv\in \Fq^n$ such that
  $\cv \star \Coi{0}{q+1} \subseteq \pu{\CC}{0}$
  has in general dimension $4$ over $\Fq$. Moreover, with a high probability,
  this space has only $2$ elements with full weight, namely,
  the vector $\xvzero^{-(q+1)}$ (which has clearly full weight) and the
  all--one vector $\onev$.
\end{itemize}
After these two steps, we know
$\xv^{q+1}$ which is unsufficient to deduce directly
$\xv$. However, we can re--apply Steps 1 and 2
replacing position $0$ by position $1$. By this manner, 
we compute $\Coi{1}{q+1}$ and then solve a problem
of the same form as (\ref{eq:compNx}) which yields
$(\xv - \onev)^{q+1}$.
\begin{itemize}
\item {\bf Step 3.} Apply Lemma~\ref{lem:MiniPol}
to get
the minimal polynomial of every entry $x_i$ of the support
$\xv$. Now, the support is known up to Galois action.

\item {\bf Step 4.} One chooses an arbitrary support $\xv'$ such
that for all $i$, 
$x_i$ and $x_i'$ have the same minimal polynomial. That is, for all $i$
either $x_i' = x_i$ or $x_i' = x_i^q$.
Since $\CC = \uv \star \Alt{r(q+1)}{\xv}{\onev}$ (Theorem~\ref{thm:properties_superGoppa}), for some $\uv \in
(\Fq^{\times})^n$, there exist a diagonal matrix 
$\mat{D}$ and a permutation matrix $\mat{P}$ such that
$$\CC = \Alt{r(q+1)}{\xv}{\onev} \mat{D} \mat{P}.$$
The permutation $\mat{P}$ is the permutation that sends 
$\xv$ onto $\xv'$. Since it arises from Galois action,
it is a product of transpositions with disjoint supports
and the supports are known. Therefore, the matrix
$\mat{D} \mat{P}$ is sparse and we know precisely
the positions of the possible nonzero entries.
The number of these unknown entries is $\approx 2q^2$
and  the linear problem
$$
\CC = \Alt{r(q+1)}{\xv}{\onev} \mat{M}
$$
whose unknowns the possible nonzero entries of $\mat{M}$
has more equations than unknowns and provide easily the
matrix $\mat{D}\mat{P}$. From them we recover $\xv$
and we have
$$\CC = \Alt{r(q+1)}{\xv}{\onev} \mat{D}.$$
This concludes the attack.
\end{itemize}

\paragraph{\bf The general case}
When the support is not full,
the main difficulty is that the resolution of~(\ref{eq:compNx})
provides $q^2 - n$ other full-weight solutions.
Thus we have $q^2-n+1$ candidates for $\xv^{q+1}$
and $q^2-n+1$ for $(\xv - \onev)^{q+1}$. A method
is explained in Appendix~\ref{ss:bypairs} which permits to gather candidates by 
pairs $(\av, \bv)$ where $\av$ is a candidate for $\xv^{q+1}$
and $\bv$ a candidate for $(\xv - \onev)^{q+1}$.
The good pair $(\xv^{q+1}, (\xv - \onev)^{q+1})$
lies among these pairs. 

\medskip

Therefore, we have to iterate Steps 3 and 4 for every
pair of candidates, which amounts to $q^2 -n$ iterations
in the worst case. 
Step 4 is also a bit more complicated in the worst case
but this has no influence on the complexity.

\begin{rem}
  Notice that the computation of the Goppa polynomial is useless to attack the scheme.
  Actually, if the secret key is a wild Goppa code $\Goppa{\xv}{\gamma^{q-1}}$,
  then 
  it is sufficient to find a pair of vectors $(\xv, \yv)$
  such that: 
  $$
  \Goppa{\xv}{\gamma^{q-1}} = \Alt{rq}{\xv}{\yv}.
  $$
  Indeed, such a representation as an alternant code allows to correct up to
  $\lfloor \frac{qr}{2} \rfloor$ errors (see Fact \ref{fa:decoding} and
  Theorem~\ref{thm:SKHN76}).
\end{rem}

\section{Limits and extensions of the attack}\label{sec:ifcompjust}
According to Lemma~\ref{lem:byPairs} and other discussions in
Appendix~\ref{sec:in_depth}, the success of the Steps 2 to 4 requires:
$$
n > 2q + 4.
$$
On the other hand, for the first step to work, the distinguisher intervals
for the computation of $\Coi{0}{2}$ up to $\Coi{0}{q-r}$ should be non 
empty, i.e. from \S\ref{sss:symmetric},
$$
{r(r+2)+2 \choose 2} > 2r(q+1)+ (q-r) -2.
$$
However, having a non empty distinguisher interval for the code
seems actually sufficient to proceed to an attack, even if there is no
distinguisher interval for the computation of $\Coi{0}{q-r}$.
Indeed, it is also possible to compute the ``negative part'' of the filtration,
i.e. the codes $\Coi{0}{-\ell}$'s for $\ell>0$. The code $\Coi{0}{-\ell}$
can be computed as
$$
\Big \{ \cv\in \Fq^n ~|~ \cv \star \Coi{0}{0} \subseteq \Coi{0}{\lfloor -\ell/2 \rfloor} \star \Coi{0}{\lceil -\ell /2 \rceil} \Big \}
$$
or, if the Schur products fill in the ambient space, several suitable
shortenings of $\Coi{0}{-\ell}$ can be computed by this manner and then
summed up to provide the whole $\Coi{0}{-\ell}$. By this manner, as soon as
we are able to compute two codes $\Coi{0}{a}$ and $\Coi{0}{b}$ such that 
$b-a = q+1$ then a statement of the form of Proposition~\ref{prop:doubleinj}
provides $\xvzero^{q+1}$. This is for instance what we did for the
$[851, 619]$ code over $\F_{32}$ presented in \S\ref{ss:implem}.
For this code it was not possible to compute $\Coi{0}{q+1}$, thus
we computed $\Coi{0}{23}$ and $\Coi{0}{-10}$.

As a conclusion, the attack or a variant by computing some $\Coi{0}{-\ell}$
may work as soon as
$$
n > 2q+4 \qquad {\rm and} \qquad {r(r+2)+2 \choose 2} > 2r(q+1)-2.
$$


\section{Complexity and Implementation}
\label{sec:example}

In what follows, by ``$\O(P(n))$''
for some function $P : \mathbb{N}\rightarrow \mathbb{R}$, we mean ``$\O(P(n))$
operations in $\Fq$''.
We clearly have $n\leq q^2$ and we also assume that $q = \O (\sqrt{n})$.

\subsection{Computation of a code product}\label{ss:costofprod}
Given two codes $\code{A}, \code{B}$ of length $n$
and respective dimensions $a$ and $b$, the computation of 
$\spc{\code{A}}{\code{B}}$ consists first in the computation
of a generator matrix of size $ab\times n$ whose
computation costs $\O (nab)$ operations.
Then the Gaussian elimination costs $\O(nab\min(n, ab))$.
Thus the cost of Gaussian elimination dominates that of
the construction step.
In particular, for a code $\code{A}$ of dimension
$k\geq \sqrt{n}$, the computation of $\code{A}^{\star 2}$
costs $\O(n^2k^2)$.
Thanks to
Proposition~\ref{prop:SolSpace}, one shows that
the dominant part of the
resolution of Problem~\ref{pr:step1},  consists in computing
$\spc{\code{A}}{\code{B}^{\bot}}$ and hence costs
$\O (na(n-b)\min (n, a(n-b)))$

\subsection{Computation of the filtration}
We first evaluate the cost of computing
$\sh{\Coi{0}{s+1}}{\Ind}$ from $\sh{\Coi{0}{s}}{\Ind}$.
The distinguisher interval described in \S\ref{ss:filtr_dist}
suggests that the dimension of $\sh{\Coi{0}{s}}{\Ind}$ used
to compute the filtration is in $\O(\sqrt{n})$.
From \S\ref{ss:costofprod}, 
the computation of the square
of $\sh{\Coi{0}{s}}{\Ind}$ costs
$\O (n^3)$ operations in $\Fq$.
Then,
the resolution of Problem~\ref{prop:SolSpace} in the 
context of Theorem~\ref{thm:fundamental}, 
costs $\O(n a (n-b)\min (n, a(n-b)))$
where $a= \dim \sh{\Coi{0}{s}}{\Ind} = \O(\sqrt{n})$
and $b = \dim \Alt{2r(q+1)+t-n+|\Ind|}{\xvzeroI}{\yv}$.
We have $n-b = \O(n)$, hence we get a cost of $\O(n^3\sqrt{n})$.

The heuristic below Proposition~\ref{prop:SolSpace}
suggests that we need to perform this computation
for $\O (\sqrt{n})$ choices of $\Ind$.
Since addition of codes is negligible compared
to $\O(n^3 \sqrt{n})$ this leads to 
a total cost of $\O (n^4)$ for the computation
of $\Coi{0}{s+1}$. This computation should be done
$q+1$ times (actually $q-r$ times from
Proposition~\ref{prop:properties_Coi}(\ref{it:stagn})
and, we assumed that $q = \O (\sqrt{n})$.
Thus, the computation of $\Coi{0}{q+1}$ costs $\O (n^4\sqrt{n})$.

\begin{rem}
  Actually, it is not necessary to compute all the terms of the filtration
  from $\Coi{0}{1}$ to $\Coi{0}{q-r}$, only $\log (q-r)$ of them are sufficient
  to get $\Coi{0}{q-r}$ since $\Coi{0}{q-r}$ is computed from
  $\Coi{0}{(q-r)/2}$. This reduces the complexity of this part to 
  $\O (n^4 \log (n))$.
\end{rem}

\subsection{Other computations}
The resolution of Problem~(\ref{eq:compNx}) in Step 2, 
costs $\O (n^4)$ (see Appendix~\ref{sec:in_depth} for further details).
The solution space
$\code{D}$ of (\ref{eq:compNx})
has $\Fq$--dimension $4$ (see Proposition~\ref{prop:dimD}
in Appendix~\ref{sec:in_depth}).
Moreover, since we are looking
for vectors of maximum weight in these solution
spaces, it is sufficient to proceed to the search
in the corresponding $3$--dimensional projective spaces.
Thus,
the exhaustive search in these solution spaces
costs $\O (q^3) = \O(n\sqrt{n})$ which is negligible.
The computation of the pairs (see \S \ref{ss:bypairs}) and that of minimal
polynomials is also negligible. Finally, the resolution
of the linear system in Step 4 costs
$\O (n^4)$ since it is very similar to Problem~\ref{pr:step1}.
Since Final step should be iterated $q^2 - n+1$ times in the worst
case, we see that the part of the attack after the computation
of the filtration costs at worst $\O (n^5)$.
Thus, the global complexity of the attack is in $\O (n^5)$
operations in $\Fq$.

\subsection{Shortcuts}
It is actually possible to reduce the complexity. Indeed, many linear
systems we have to solve have $b$ equations and $a$ unknowns with
$b \gg a$. For such systems it is possible to extract $a+\epsilon$
equations chosen at random and solve this subsystem which has the same
solution set with a high probability. This probabilistic shortcut permits
the computation of the square of a code in $\O (n^3)$ and reduces the cost
of Step 4 to $\O (n^3)$. By this manner we have an overall complexity
of $\O (n^4)$.

\subsection{Implementation}\label{ss:implem}
This attack has been implemented with
$\textsc{Magma}$ \cite{BCP97} and run over
random examples of codes corresponding
to the seven entries \cite[Table 1]{BLP10}
for which $m=2$ and $r>3$.
For all these parameters, our attack succeeded.
We summarize here the average running times
for at least $50$ random keys per $4$--tuple of parameters,
obtained with an Intel\textsuperscript{\textregistered} Xeon 2.27GHz. 

\bigbreak

\begin{center}
\begin{tabular}{@{}*{5}{c}@{}}
\toprule 
$(q,n,k,r)$ & (29,781, 516,5)  & (29, 791, 575, 4) & (29,794,529,5) &
(31, 795, 563, 4) \\ \midrule
Average time & $16$min& $19.5$min & $15.5$min & $31.5$min   \\
 \bottomrule
\end{tabular}
\end{center}

\medbreak

\begin{center}
\begin{tabular}{@{}*{4}{c}@{}}
\toprule 
$(q,n,k,r)$ & (31,813, 581,4)  & (31, 851, 619, 4) & (32,841,601,4)  \\ \midrule
Average time & $31.5$min & $27.2$min & $49.5$min  \\ 
 \bottomrule
\end{tabular}
\end{center}

\bigbreak

\begin{rem}
In the above table the code dimensions are 
not the ones mentioned in \cite{BLP10}.
What happens here is that the formula for the dimension given \cite[p.153,\S1]{BLP10} is wrong for such cases: it understimates
the true dimension for wild Goppa codes over quadratic extensions when the degree $r$ of the irreducible polynomial $\gamma$ is larger than $2$ as shown by 
Theorem~\ref{thm:properties_superGoppa}(\ref{it:rrmoins2}).
\end{rem}

All these parameters are given in \cite{BLP10} with a $128$-bit security that is measured
against information set decoding attack which is described in \cite[p.151, Information set decoding \S1]{BLP10} as the 
``\emph{top threat
against the wild McEliece cryptosystem for $\F_3$, $\F_4$, etc.}''. It should be mentioned that these
parameters   are marked in \cite{BLP10} by the biohazard symbol \Biohazard \ (together with about two dozens other parameters). This corresponds, as explained in \cite{BLP10}, to parameters for which the number of possible
monic Goppa polynomials of the form $\gamma^{q-1}$ is smaller than $2^{128}$. The authors in \cite{BLP10} choose in this case 
a support which is significantly smaller than $q^m$ ($q^2$ here)
in order to avoid attacks that fix a support of size $q^m$ and then enumerate all possible polynomials. Such attacks exploit the fact that two Goppa codes of length $q^m$ with the same polynomial  are \emph{permutation equivalent}.
We recall that the \emph{support-splitting} algorithm \cite{S00}, when applied to permutation equivalent codes, generally finds in polynomial time a permutation that sends one code onto the other. 
The authors of \cite{BLP10} call this requirement on the length the {\em second defense} and 
write \cite[p.152]{BLP10}.

``{\em The strength of the second defense is unclear:
we might be the first to ask whether the support-splitting idea can be generalized to handle many sets 
$\{a_1,\dots,a_n\}$}
\footnote{$\{a_1,\dots,a_n\}$ means here the support of the Goppa code.}
{\em simultaneously, and we would not be surprised if the answer turns out to be yes.}''
 The authors also add in \cite[p.154,\S1]{BLP10} that ``{\em the security
 of these cases}\footnote{meaning here the cases marked with \Biohazard.}
{\em depends on the strength of the second defense discussed in Section 6}''.
We emphasize that our attack has nothing to do with the strength or a potential weakness of the second defense. Moreover, it does not exploit at all the fact that there
are significantly less than $2^{128}$ Goppa polynomials. 
This is obvious from the way our attack works and this can also be verified
by attacking parameters which were not proposed in \cite{BLP10} 
but for which there are more than $2^{128}$ monic wild Goppa polynomials to
check. As an illustration, we are also able to recover the secret key in an
average time of 24 minutes when the public key is a code over $\F_{31}$, of
length $900$ and with a Goppa polynomial of degree $14$.
In such case, the number of possible Goppa polynomials is larger than $2^{134}$
and 
according to Theorem~\ref{thm:properties_superGoppa}, the public key
has parameters $[n=900, k\geq 228, d\geq 449]_{31}$.
Note that
security of such a key with respect to information set decoding \cite{P10} is also high (about $2^{125}$ for such parameters).


\section{Conclusion}
\label{sec:conclusion}

The McEliece scheme based on Goppa codes has withstood all cryptanalytic attempts up to now, even if a related system based on GRS codes
\cite{N86} was successfully attacked in \cite{SS92}. Goppa codes are subfield subcodes of GRS codes and it was advocated that taking
the subfield subcode hides a lot about the structure of the underlying code and also makes these codes more random-like. This is sustained by the 
fact that the distance distribution becomes indeed random \cite{MS86} by this operation whereas GRS codes behave differently from random
codes with respect to this criterion. This attack presented at the conference
EUROCRYPT 2014 was the first example of a cryptanalysis which questions
this belief by providing an algebraic cryptanalysis
which is of polynomial complexity and which applies to many ``reasonable parameters'' of a McEliece scheme when the Goppa code is the $\Fq$-subfield subcode
of a GRS code defined over  $\Fqq$.
Subsequently to our attack,
this uncertainty on the security of code based cryptosystems using wild
Goppa codes  has been strengthened by another
cryptanalysis based on the resolution of a system of multivariate polynomial
equations \cite{FPP14}.

It could be argued that our attack applies to a rather restricted class of Goppa codes, namely wild Goppa codes of extension degree two.
This class of codes also presents certain peculiarities as shown by
Theorem~\ref{thm:properties_superGoppa}
which were helpful for mounting an attack.
However, it should be pointed out that the crucial ingredient which made this attack possible is the fact that such codes
could be distinguished from random codes by square code considerations. A certain filtration of subcodes was indeed exhibited here
and it turns out that  shortened versions of these codes were related together by the star product. This allowed to reconstruct the filtration
and from here the algebraic description of the Goppa code could be recovered. The crucial point here is really the existence of such a 
filtration whose elements  are linked together by the star product. The fact that these codes were linked together by the star product
is really related to the fact that the square 
code of certain shortened codes of the public code were of unusually low dimension which is precisely the fact that yielded the
aforementioned distinguisher. This  raises the issue whether other families of Goppa codes or alternant codes which 
can be distinguished from random codes by such square considerations \cite{FGOPT13} can be attacked by techniques of this kind.
This covers high rate Goppa or alternant codes, but also other Goppa or alternant codes when the degree of extension is equal to $2$.
All of them can be distinguished from random codes by taking square codes of a shortened version of the dual code.


\bibliographystyle{alpha}
\bibliography{codecrypto}

\appendix
\section{Reducing to the case $x_0=0$, $x_1=1$}\label{ss:x0x1}
The fact that we can choose $x_0$ to be equal to $0$ and $x_1$
to be equal to $1$ for the support $\xv$ of $\CC$ 
 follows at once from the following lemma together with
the $2$-transitivity of the affine maps $x \mapsto a x+b$ over $\Fqm$.
This lemma is basically folklore, but since we did not find a reference
giving this
lemma in exactly this form we have also provided a proof for it.
\begin{lemma}\label{lem:2transitive}
Consider a Goppa code $\Goppa{\xv}{\Gamma(x)}$ defined over $\Fq$ and of extension degree $m$.
Let $a,b \in \Fqm$ with $a\neq 0$ and let $\psi(z) \eqdef az +b$.
We have $\Goppa{\xv}{\Gamma(z)} = \Goppa{\psi(\xv)}{\Gamma(\psi^{-1}(z))}.$
\end{lemma}

\begin{proof}{}
We first observe that for any alternant code of some length $n$, degree $r$, extension degree $m$, defined over 
$\fq$ we have

\begin{equation}\label{eq:alternant}
\Alt{r}{\xv}{\yv}  =  \Alt{r}{a\xv+b}{\yv}. 
\end{equation}
This can be verified as follows.
Let $\cv=(c_i)_{0 \leq i \leq n-1}$ be a codeword in
 $\Alt{r}{\xv}{\yv}$. We are going to prove that it also belongs to 
 $\Alt{r}{a\xv+b}{\yv}$. It suffices to prove that for any polynomial 
 $P$ in $\F_{q^m}[X]$ of degree at most $r-1$ we have 
 $\sum_{i=0}^{n-1} c_i y_i P \left(ax_i+b \right)=0$.
 In order to prove this, let us first observe that 
 we may write 
 $P(ax+b)$ as a polynomial $Q(x)$ of degree at most $r-1$ which depends
on $a$ and $b$. This implies that 
$$
 \sum_{i=0}^{n-1} c_i y_i P \left(ax_i+b \right)=\sum_{i=0}^{n-1} c_i y_i Q(x_i)=0.
$$
where the last equality follows from the definition of $\Alt{r}{\xv}{\yv}$.
In other words, we have just proved that
$\cv \in \Alt{r}{a\xv+b}{\yv}$. This proves that
\begin{equation}
\label{eq:inclusion}
\Alt{r}{\xv}{\yv}  \subseteq  \Alt{r}{a\xv+b}{\yv}. 
\end{equation}
The inclusion in the other direction by observing that by  using 
\eqref{eq:inclusion} on $\Alt{r}{a\xv+b}{\yv}$ with the affine map $\Psi^{-1}$ we obtain
$$
\Alt{r}{a\xv+b}{\yv}  \subseteq \Alt{r}{\Psi^{-1}(a\xv+b)}{\yv}=\Alt{r}{\xv}{\yv}
$$
and this proves \eqref{eq:alternant}.
This is used to finish the proof of Lemma \ref{lem:2transitive} by observing that
\begin{eqnarray*}
\Goppa{\xv}{\Gamma(z)} & = & \Alt{r}{\xv}{\yv}\\
& = & \Alt{r}{a \xv + b}{\yv}\\
& = & \Goppa{a\xv+b}{\Gamma(\psi^{-1}(z))}
\end{eqnarray*}
where $r=\deg \Gamma$ and $\yv= {\Gamma(\xv)}^{-1} \cdot$
\end{proof}


\section{Further details on the behaviour
of the filtration $\big(\Coi{0}{s}\big)_{s\in \Z}$}\label{sec:stagn}

The aim of this appendix is to give a complete proof of
Proposition~\ref{prop:properties_Coi}. For convenience, let us remind its statement.

\medbreak

\noindent {\bf Proposition~\ref{prop:properties_Coi}}
{\em Under Assumption \ref{ass:fundamental} (i), we have
\begin{enumerate}[(i)]
  \item $\Coi{0}{1} = \sh{\CC}{0}$;
  \item $\Coi{0}{0} = \pu{\CC}{0}$;
  \item $\forall s\in \Z,\ \dim \Coi{0}{s}
    -\dim \Coi{0}{s+1} \leq 2$;
  \item $\Coi{0}{q-r} = \Coi{0}{q+1}$;
  \item $\forall s\in \Z,\ \Coi{0}{s} = 
    \Alt{r(q+1)+s-1}{\xvzero}{\yvzero}$ for 
    $$
    \yvzero \eqdef \gamma^{-(q+1)}(\xvzero) \star \xvzero^{-(s-1)}.
    $$
\end{enumerate}
where we recall that $\xv_0$ denotes the vector
$\xv$ punctured at position $0$ and that $r$ denotes the
degree of $\gamma$.}

\medbreak

The following lemma is useful in the proofs to follow.

\begin{lem}\label{lem:locator_short}
For all $i \in \{1, \ldots , n-1\}$, we have $\loc{\xv}'(x_i) = x_i \loc{\xvzero}'(x_i)$. Equivalently,
  $$
  \loc{\xv}'(\xvzero) = \spc{\xvzero}{ \loc{\xvzero}'(\xvzero)}.
  $$
\end{lem}

\begin{proof}
Recall that  $x_0 = 0$ and therefore we have:
  $$
  \loc{\xv}(z) = \prod_{i=0}^{n-1} (z- x_i) =
  z \prod_{i=1}^{n-1} (z-x_i) =
  z \loc{\xvzero}(z).
  $$
  Therefore, 
  $
  \loc{\xv}'(z) = z\loc{\xvzero}'(z) + \loc{\xvzero}(z)
  $
  and since $\loc{\xvzero}(x_i) = 0$ for all
  $i \in \{1, \ldots , n-1\}$, we get the result.
\end{proof}

\subsection{Proof of  (\ref{it:desc_as_alt}) and further results about the structure of the $\Coi{0}{s}$'s }

We will start by proving \eqref{it:desc_as_alt}.
Let $s\in \Z$. By definition
$$
\Coi{0}{s} \eqdef  
\cset{\frac{\gamma^{q+1}(x_i)}{\loc{\xv}'(x_i)} x_i^s f(x_i)}
{i \in \{1,\dots,n-1\}}{f \in \Fqq[z]_{<n- r(q+1)-s}} 
\cap \Fq^{n-1}.
$$
Then, from Lemma \ref{lem:locator_short}, we have
\begin{align*}
\Coi{0}{s}  & = 
\cset{\frac{\gamma^{q+1}(x_i)}{x_i\loc{\xvzero}'(x_i)} x_i^s f(x_i)}
{i \in \{1,\dots,n-1\}}{f \in \Fqq[z]_{<n- r(q+1)-s}} 
\cap \Fq^{n-1}.\\
            & =
\cset{\frac{\gamma^{q+1}(x_i)}{\loc{\xvzero}'(x_i)} x_i^{s-1} f(x_i)}
{i \in \{1,\dots,n-1\}}{f \in \Fqq[z]_{<(n-1)- r(q+1)-(s-1)}} 
\cap \Fq^{n-1}.
\end{align*}
The very definition of alternant codes (Definition~\ref{def:subfield_subcode})
yields
$$\Coi{0}{s} = \Alt{r(q+1)+s-1}{\xvzero}{\yvzero},\quad
{\rm where}\quad \yvzero = \gamma^{-(q+1)}(\xvzero) \star \xvzero^{-(s-1)}.
$$

It should be noted that the $\Coi{0}{s}$'s  can also be viewed  as Goppa codes twisted by multiplying the positions
by some fixed constants as explained by the following proposition which sheds some further light on the structure of the 
codes $\Coi{0}{s}$:
\begin{proposition}\label{prop:equiv_BCH}
  Let 
  $
  \uv \eqdef \gamma^{q+1}(\xvzero) \star \xvzero^{-r(q+1)}
  $
  where $r$ denotes the degree of the polynomial
  $\gamma$ such that $\CC = \Goppa{\xv}{\gamma^{q+1}}$.
  Then $\uv \in \Fq^{n-1}$ and for $s \geq -r(q+1)$:
  $$
  \Coi{0}{s} = \uv \star \Goppa{\xvzero}{z^{r(q+1)+s-1}}.
  $$
\end{proposition}

\begin{proof}
  It has been discussed in \S~\ref{ss:attack:key_tools},
  that $(q+1)$--th powers in $\Fqq$ are norms and hence are elements
  of $\Fq$. Therefore, $\uv \in \Fq^{n-1}$.
  Now, recall the definition of $\Coi{0}{s}$ as
  $$
    \Coi{0}{s} \eqdef
      \cset{\frac{\gamma^{q+1}(x_i)}{\loc{\xvzero}'(x_i)} x_i^{s-1} f(x_i)}{
      i \in \{1, \ldots, n-1\}}{f \in \Fqq[z]_{<n-r(q+1)-s}}
      \cap \Fq^{n-1}.
  $$
  Since $\uv = \gamma^{q+1}(\xvzero) \star \xvzero^{-r(q+1)}$
  is a vector with entries in $\Fq$,
  it can get in the subfield subcode and
  $$
  \uv^{-1} \star \Coi{0}{s}
  = \cset{\frac{x_i^{r(q+1)}}{\loc{\xvzero}'(x_i)}x_i^{s-1}f(x_i)}{
  i \in \{1, \ldots, n-1\}}{f\in \Fqq[z]_{<n-1-r(q+1)-(s-1)}} \cap \Fq^{n-1}.
  $$
  Finally, thanks to the description of Goppa codes as evaluation codes
  in Lemma~\ref{lem:descr_Goppa_as_eval}, the right hand term of the
  above equality is $\Goppa{\xvzero}{z^{r(q+1)+s-1}}$,
  which concludes the proof.
\end{proof}

\subsection{Proof of (\ref{it:short_Coi})}
From (\ref{it:desc_as_alt}), applied to $s=1$, we have
$$
\Coi{0}{1} = \Alt{r(q+1)}{\xvzero}{\gamma^{-(q+1)}(\xvzero)}.
$$
Thus, the very definition of Goppa codes (Definition~\ref{def:ClassicalGoppa})
entails
$$
\Coi{0}{1} = \Goppa{\xvzero}{\gamma^{q+1}}.
$$
Therefore, Corollary~\ref{cor:shortened_Goppa} on 
shortened Goppa codes asserts that
$
\Coi{0}{1} = \sh{\CC}{0}.
$

\subsection{Proof of (\ref{it:codim_Coi})}

Let us bring in the family of GRS codes
${(\code{F}_s)}_{s\in \Z}$ defined as
\begin{equation}\label{eq:defF}
\code{F}_s
\eqdef \cset{\frac{\gamma^{q+1}(x_i)}{\loc{\xvzero}'(x_i)} x_i^{s-1} f(x_i)}
{i \in \{1,\dots,n-1\}}{f \in \Fqq[z]_{<n- r(q+1)-s}} \cdot
\end{equation}
We have
$$
\forall s\in \Z,\ \Coi{0}{s} = \code{F}_s \cap \Fq^{n-1}.
$$
Moreover, it is readily seen that for all $s$, $\code{F}_{s+1} \subseteq \code{F}_s$
and $\dim {\code{F}_s} - \dim \code{F}_{s+1} \leq 1$, with equality
if $\code{F}_s$ is nonzero.
Then, the proof of (\ref{it:codim_Coi}) is a direct consequence of
the following lemma.

\begin{lem}
  Let $\code{A}, \code{A}' \subseteq \Fqm^n$ be two codes such that
  $\code{A} \subseteq \code{A}'$.
  Then
  $$
  \dim_{\Fq} (\code{A}'\cap \Fq^n) - \dim_{\Fq} (\code{A} \cap \Fq^n) \leq 
  m\left(\dim_{\Fqm}(\code{A}') - \dim_{\Fqm} (\code{A})\right).
  $$
\end{lem}

\begin{proof}
  Thanks to Delsarte Theorem (Theorem~\ref{th:Delsarte}) it is equivalent
  to prove that
  $$
  \dim_{\Fq} \tr \left(\code{A}^{\bot}\right) - \dim_{\Fq} \tr \left(
    {\code{A}'}^{\bot} \right)
  \leq m
  \left(\dim_{\Fqm} \code{A}^{\bot} - \dim_{\Fqm} {\code{A}'}^{\bot} \right).
  $$
  To prove it, choose any code $\code{B} \subseteq \Fqm^n$ such that
  $\code{A}^{\bot} = {\code{A}'}^{\bot} \oplus \code{B}$.
  Then, we clearly have
  $$
  \dim_{\Fq}
  \tr \left(\code{A}^{\bot}\right) \leq \dim_{\Fq} \tr \left( {\code{A}'}^{\bot}
  \right) + \dim_{\Fq} \tr \left( \code{B} \right).
   $$
  Finally, from \cite[\S~7.7]{MS86}, we get
  \begin{align*}
    \dim_{\Fq} \tr \code{B} &\leq m \dim_{\Fqm} \code{B} \\
                           &\leq m \left( \dim_{\Fqm} \code{A}^{\bot} -
                             \dim_{\Fqm} {\code{A}'}^{\bot}\right),
  \end{align*}
  which yields the result.
\end{proof}

\subsection{Proof of (\ref{it:simplify})}
Statement (\ref{it:simplify}) is less obvious than it looks like and
far less obvious than (\ref{it:short_Coi}).
Indeed, let
$$
\code{F}  \eqdef \cset{\frac{\gamma^{q+1}(x_i) f(x_i)}{\loc{\xv}'(x_i)}}{0\leq i<n}{f \in \Fqq[z]_{<n - r(q+1)}}
\cdot
$$
Using Lemma~\ref{lem:locator_short}, one proves that
$\code{F}_0 = \pu{\code{F}}{0}$
where $\code{F}_0$ is 
defined in (\ref{eq:defF}).
Moreover, we have: 
$$
\CC = \code{F} \cap \Fq^n
\quad \textrm{and}
\quad
\Coi{0}{0} = \code{F}_0 \cap \Fq^{n-1}. 
$$

Therefore, 
$$
\pu{\CC}{0} = \pu{\code{F}\cap \Fq^n}{0}
\quad \textrm{while} \quad
\Coi{0}{0} = \pu{\code{F}}{0} \cap \Fq^{n-1}.$$
Hence, from Proposition~\ref{prop:commutation}, we get
\begin{equation}
  \label{eq:inclusion_a_priori}
  \pu{\CC}{0} \subseteq \Coi{0}{0}
\end{equation}
and there is \textit{a priori} no reason for the converse inclusion to be true.
We will prove this by first observing that
\begin{proposition}
\begin{equation}
  \label{eq:incl_sequence2}
    \dim \Coi{0}{0} - \dim \Coi{0}{1} \geq 1.
\end{equation}
\end{proposition}

\begin{proof}
First of all, notice that $\pu{\CC}{0}=\CC$ because $\CC$ is of minimum distance
$>1$. Assumption \ref{ass:fundamental} (i) tells us that 
$\sh{\CC}{0} \neq \CC$ and therefore
$\sh{\CC}{0}$ is strictly included in $\pu{\CC}{0}$.
In summary,
thanks to (\ref{it:short_Coi}) and (\ref{eq:inclusion_a_priori}), we have
\begin{equation}
   \label{eq:incl_sequence}
\underbrace{\Coi{0}{1}}_{= \sh{\CC}{0}} \varsubsetneq \pu{\CC}{0} \subseteq \Coi{0}{0}
\end{equation}
and hence,
\begin{equation}
  \label{eq:incl_sequence2}
    \dim \Coi{0}{0} - \dim \Coi{0}{1} \geq 1.
\end{equation}
\end{proof}

On the other hand, we can bound from above this difference of dimensions.
This follows from

\begin{proposition}\label{cor:codim_leq1}
  We have
  $$
  \dim \Coi{0}{1} - \dim \Coi{0}{-r} \leq 1.
  $$
\end{proposition}

\begin{proof}
From Proposition \ref{prop:equiv_BCH}, we have
\begin{align*}
\Coi{0}{1} = \uv \star \Goppa{\xvzero}{z^{r(q+1)}} \quad {\rm and} \quad
\Coi{0}{-r} & = \uv \star \Goppa{\xvzero}{z^{r(q+1) - r-1}} \\
            & = \uv \star \Goppa{\xvzero}{z^{rq-1}}.
\end{align*}
From Theorem~\ref{thm:SKHN76} and Remark~\ref{rem:wild_Goppa},
we have 
$$
\Goppa{\xvzero}{z^{rq- 1}} = \Goppa{\xvzero}{z^{rq}}.
$$
In addition, from \cite[Theorem 4]{COT13}, we have
$$
\dim \Goppa{\xvzero}{z^{rq}} - \dim \Goppa{\xvzero}{z^{r(q+1)}} \leq 1.
$$
This yields the result.
\end{proof}

\paragraph{\bf Conclusion} Putting inclusion sequence (\ref{eq:incl_sequence})
in the filtration of the $\Coi{0}{j}$'s, we get the inclusion sequence
$$
\Coi{0}{1} \varsubsetneq \pu{\CC}{0} \subseteq \Coi{0}{0}
\subseteq \Coi{0}{-1} \subseteq \cdots \subseteq \Coi{0}{-r}
$$
Using Proposition \ref{cor:codim_leq1}, we prove that,
in the above inclusion sequence, every
inclusion is an equality but the left-hand one.
In particular,
$$
\pu{\CC}{0} = \Coi{0}{0},
$$
which concludes the proof of (\ref{it:simplify}).

\medskip

Actually, we got other deep results
namely, $\Coi{0}{0} = \Coi{0}{-r}$ and $\dim \Coi{0}{0} - \dim
\Coi{0}{1} \leq 1$. Using Proposition~\ref{prop:equiv_BCH},
we obtain an interesting result on Goppa codes which
 clarifies  \cite[Theorem 4]{COT13}.
\begin{cor}
  \label{cor:truc_Goppa}
  Let $\ell$ be a positive integer
  and $\xv \in \Fqq^n$ be a support, then:
  \begin{enumerate}[(i)]
     \item\label{it:truc_goppa_eq} $\Goppa{\xv}{z^{\ell q - 1}} =
       \Goppa{\xv}{z^{\ell (q+1)-1}}$;
     \item $\dim \Goppa{\xv}{z^{\ell (q+1)-1}} -
       \dim \Goppa{\xv}{z^{\ell (q+1)}} \leq 1$.
  \end{enumerate}
\end{cor}

\subsection{Proof of (\ref{it:stagn})}
Thanks to Proposition~\ref{prop:equiv_BCH}, it reduces to
prove that
$$
\Goppa{\xvzero}{z^{(r+1)(q+1)-1}} = \Goppa{\xvzero}{z^{(r+1)q-1}},
$$
which is a direct consequence of
Corollary~\ref{cor:truc_Goppa}(\ref{it:truc_goppa_eq})
in the case $\ell = r+1$.


\section{An in--depth presentation of the attack}
\label{sec:in_depth}

Here we give a complete presentation of the attack in the general case, 
i.e. for a possibly non full support $\xv$. As explained in \S\ref{sec:attack},
the attack divides into four steps:
\begin{itemize}
   \item {\bf Step 1.} Compute the terms of the filtrations ${\big(\Coi{0}{j} \big)}_j$ 
     and ${\big(\Coi{1}{j}\big)}_j$ up to $\Coi{0}{q+1}$ and $\Coi{1}{q+1}$, using
     the methods presented in \S\ref{sec:nested}.
   \item {\bf Step 2.} Compute $\xv^{q+1}$ and $(\xv - \onev)^{q+1}$ thanks to 
     Proposition~\ref{prop:doubleinj}.
   \item {\bf Step 3.} Compute the minimal polynomials of every entry $x_i$ of the support $\xv$ using Lemma~\ref{lem:MiniPol}.
   \item {\bf Step 4.} Compute a matrix $\mat{M}$ solution of the linear problem
     $$
     \CC = \CC' \mat{M}
     $$
     where $\mat{M}$ is a matrix with many prescribed zero entries
     and $\CC' = \Alt{r(q+1)}{\xv}{\onev}$ and obtain from $\mat{M}$
     the whole structure of $\CC'$.
\end{itemize}

Step 1 is explained in depth in \S\ref{sec:nested} and Step 3
is straightforward (it is a direct application of Lemma~\ref{lem:MiniPol}). 
Thus, in this appendix, we give further details on Steps 2 and 4.

\subsection{Further details on Step 2 of the attack}\label{ss:furtherStep2}
As explained in \S\ref{sec:attack}, the computation of $\xv^{q+1}$
or, more precisely, that of $\xvzero^{-(q+1)}$ reduces to 
solving 
Problem~(\ref{eq:compNx}) which we recall here:
\begin{equation}\tag{\ref{eq:compNx}}
    \left\{
      \begin{array}{l}
        \cv \star \Coi{0}{q+1} \subseteq \pu{\CC}{0} \\ \\
        \forall i \geq 1,\ c_i \neq 0,\quad {\rm (i.e.}\ \cv \ {\rm has\ full\ weight)} \\ \\
        c_1 = 1.
      \end{array}
    \right.
\end{equation}

Remind that, from Proposition~\ref{prop:properties_Coi}(\ref{it:simplify}),
we know that $\pu{\CC}{0} = \Coi{0}{0}$. 
Then, according to Proposition~\ref{prop:SolSpace}, the subspace of vectors
$c\in \Fq^{n-1}$ such that $\cv \star \Coi{0}{q+1} \subseteq \pu{\CC}{0}$
is the space
$$
\DC \eqdef {\left( \spc{\Coi{0}{q+1}}{\Coi{0}{0}^{\bot}}\right)}^{\bot}.
$$
We will first investigate the structure of $\DC$ and in particular its
dimension. Then, we will study its set of full weight codewords.
For this sake we will use repeatedly the following elementary lemma.

\begin{lem}\label{lem:stardual}
  Let $\code{A} \subseteq \Fq^n$ be a code and $\uv \subseteq {(\Fq^{\times})}^n$,
  then
  $$
  {(\uv \star \code{A})}^{\bot} = \uv^{-1} \star (\code{A}^{\bot}).
  $$
\end{lem}

\begin{proof}
  Since $\uv$ is invertible, then, clearly, both codes have the same dimension
  and it is sufficient to prove inclusion ``$\supseteq$''.
  Let $\av \in \code{A}$ and $\bv \in \code{A}^{\bot}$, then
  $$
  \langle \uv \star \av, \uv^{-1} \star \bv \rangle =
  \sum_i u_i a_i u_i^{-1} b_i = \sum_i a_i b_i = 
  \langle \av, \bv \rangle = 0.
  $$
  This concludes the proof.
\end{proof}

\subsubsection{The structure of the code $\DC$}
We start with a rather technical statement which is fundamental in what 
follows.

\begin{proposition}\label{prop:structureD}
  We have
  $$
  \xvzero^{-(q+1)} \star \left(\RS{q+2}{\xvzero} \cap \Fq^{n-1} \right)
  \subseteq \code{D}.
  $$
\end{proposition}

\begin{proof}
  First let us rewrite the codes $\Coi{0}{0}$ and $\Coi{0}{q+1}$
  in a more convenient way. By
  definition
  \begin{align*}
  \Coi{0}{q+1} & =
  \cset{ \frac{\gamma^{q+1}(x_i)}{\loc{\xv}'(x_i) } x_i^{q+1} f(x_i)}
  {i \in \{1, \ldots, n-1\}}{
  f\in \Fqq[z]_{<n - (r+1)(q+1)}} \cap \Fq^{n-1}   
  \end{align*}
  In the very same way as in the proof of Proposition~\ref{prop:equiv_BCH},
  since the $(q+1)$--th powers are norms and hence are in $\Fq$, they can
  get out of the subfield subcode.
  \begin{align*}
    \Coi{0}{q+1} &=
    \gamma^{q+1}(\xvzero) \star \xvzero^{q+1} \star
    \cset{\frac{1}{\loc{\xv}'(x_i)} f(x_i)}{i \in \{1, \ldots, n-1\}}{
    f \in \Fqq[z]_{< n - (r+1)(q+1) }} \cap {\Fq^{n-1}}.
  \end{align*}
 Since the codes have length $n-1$, it is more relevant to write 
 $\Fqq[z]_{< n - (r+1)(q+1) }$ as\\ $\Fqq[z]_{< (n-1) - (r+1)(q+1) + 1}$.
  Then, thanks to Lemma~\ref{lem:locator_short}, we get
  \begin{align*}
    \Coi{0}{q+1} &=
    \gamma^{q+1}(\xvzero) \star \xvzero^{q+1} \star \\
    & \qquad \cset{\frac{1}{x_i\loc{\xvzero}'(x_i)} f(x_i)}{i \in \{1, \ldots, n-1\}}{
    f \in \Fqq[z]_{< (n-1) - (r+1)(q+1) + 1}} \cap \Fq^{n-1}.
  \end{align*}
   Consequently, by the description of alternant codes as evaluation codes
  (Lemma~\ref{lem:decription_as_eval}), we obtain
  \begin{equation}\label{eq:new_C0qp1}
               \Coi{0}{q+1} = \gamma^{q+1}(\xvzero) \star \xvzero^{q+1} \star
                 \Alt{(r+1)(q+1)-1}{\xvzero}{\xvzero}.
  \end{equation}
  In the very same manner, we prove that
  \begin{equation}\label{eq:newCoiO}
  \Coi{0}{0} =  \gamma^{q+1}(\xvzero)  \star
                 \Alt{r(q+1)-1}{\xvzero}{\xvzero}.
  \end{equation}
  From the definition of alternant codes
  (Definition~\ref{def:subfield_subcode}) together with Delsarte
  Theorem (Theorem~\ref{th:Delsarte}) and Lemma~\ref{lem:stardual}, we get
  \begin{equation}\label{eq:dualCoi0}
  \Coi{0}{0}^{\bot} = \gamma^{-(q+1)}(\xvzero) \star \tr \left(
     \GRS{r(q+1)-1}{\xvzero}{\xvzero}\right)
  \end{equation}
  From (\ref{eq:new_C0qp1}) and (\ref{eq:dualCoi0}), 
  $$
  \spc{\Coi{0}{q+1}}{\Coi{0}{0}^{\bot}} = 
  \xvzero^{q+1} \star \spc{\Alt{(r+1)(q+1)-1}{\xvzero}{\xvzero}}{
  \tr \left( \GRS{r(q+1)-1}{\xvzero}{\xvzero}\right)}.
  $$
  Since the alternant code is defined over $\Fq$, it can get in the trace:
  \begin{equation}\label{eq:prodCoi0Coiqp1}
  \spc{\Coi{0}{q+1}}{\Coi{0}{0}^{\bot}} = 
  \xvzero^{q+1} \star \tr \left(\spc{\Alt{(r+1)(q+1)-1}{\xvzero}{\xvzero}}{
   \GRS{r(q+1)-1}{\xvzero}{\xvzero}}\right).
  \end{equation}
  By definition of alternant codes (Definition~\ref{def:subfield_subcode})
  and by duality for GRS codes (Proposition~\ref{prop:dualGRS}),
  $$
    \Alt{(r+1)(q+1)-1}{\xvzero}{\xvzero} = 
    \GRS{n-(r+1)(q+1)}{\xvzero}{\xvzero^{-1} \star {\loc{\xvzero}'}^{-1}(\xvzero)}
    \cap \Fq^{n-1}
  $$
  Therefore, since every code contains its subfield subcode, 
  \begin{equation}
    \label{eq:incl_always_eq}
  \begin{aligned}
  \Alt{(r+1)(q+1)-1}{\xvzero}{\xvzero} & \star
   \GRS{r(q+1)-1}{\xvzero}{\xvzero} \subseteq  \\
   & \spc{\GRS{n-(r+1)(q+1)}{\xvzero}{\xvzero^{-1} \star {\loc{\xvzero}'}(\xvzero)^{-1}}}{\GRS{r(q+1)-1}{\xvzero}{\xvzero}}.
  \end{aligned}  
  \end{equation}
  Thus, from Proposition~\ref{prop:GRSsquare}(\ref{cas1:grssquare}) on 
  products on GRS codes,
  we get
  \begin{equation}\label{eq:horrible_product}
  \Alt{(r+1)(q+1)-1}{\xvzero}{\xvzero} \star
   \GRS{r(q+1)-1}{\xvzero}{\xvzero} \subseteq  
   \GRS{(n-1)-(q+2)}{\xvzero}{\loc{\xvzero}'(\xvzero)^{-1}}.
  \end{equation}
  Equations (\ref{eq:prodCoi0Coiqp1}) and (\ref{eq:horrible_product}) yield
  $$
    \spc{\Coi{0}{q+1}}{\Coi{0}{0}^{\bot}} \subseteq \xvzero^{q+1} \star
      \tr \left( \GRS{(n-1)-(q+2)}{\xvzero}{\loc{\xvzero}'(\xvzero)^{-1}}
        \right).
  $$
  By dualizing and thanks to Lemma~\ref{lem:stardual}, to Delsarte Theorem
  (Theorem~\ref{th:Delsarte}) and Proposition~\ref{prop:dualGRS},
  we get
  \begin{align*}
  \code{D} 
  & \supseteq  \xvzero^{-(q+1)} \star \left(\RS{q+2}{\xvzero} \cap
    \Fq^{n-1}\right).
  \end{align*}
\end{proof}

\bigbreak

\paragraph{\bf Discussion on the equality}
While Proposition~\ref{prop:structureD} is only an inclusion, it turns
out that in all our experiments, the inclusion was an equality.
It is worth nothing that the reason why this equality typically holds is 
more or less the reason why our distinguisher works.

Indeed, for the equality to hold, (\ref{eq:incl_always_eq}) should be an
equality. The right-hand product in  (\ref{eq:incl_always_eq})
is a GRS code of dimension $(n-1) - (q+1)$ (see~(\ref{eq:horrible_product})),
while the left hand one is a product of codes of respective (designed)
dimensions $n - 2(r+1)(q+1)$ and $r(q+1)-1$. From Proposition~\ref{pr:typ_dim},
the product of two random codes with these dimensions would be
$$\min \{n-1, (n - 2(r+1)(q+1))(r(q+1)-1)\}$$
For cryptographic sizes of parameters, the above min is $n-1$ and hence, with
a very high probability, the left-hand product in (\ref{eq:incl_always_eq})
fills in the right-hand one. This explains, why the inclusion in 
Proposition~\ref{prop:structureD} is almost always an equality.

\medbreak

Let us now investigate further the structure of the code
$\RS{q+2}{\xvzero} \cap \Fq^{n-1}$.

\begin{nota}
  Let $\alpha$ be a primitive element of $\Fqq/\Fq$.
  In what follows we denote by $\code{E}$ the following code which
  is used repeatedly
  $$
  \EC \eqdef {\langle \onev, \tr(\xvzero), \tr(\alpha \xvzero),
    N(\xvzero) \rangle}_{\Fq}.
  $$
\end{nota}

\begin{proposition}\label{prop:dimD}
  We have,
  $$\EC \subseteq
    \RS{q+2}{\xvzero} \cap \Fq^{n-1},
  $$
  with equality when the support $\xv$ is full.
\end{proposition}

\begin{proof}
We first prove the result under the assumption that $\xv$ is full.
Our goal is to describe the polynomials $h$ in $\Fqq[x]_{<q+2}$
satisfying
$$
\forall x\in \Fqq^{\times},\ h(x) = h(x)^q.
$$
Or equivalently,
\begin{equation}
  \label{eq:heqhq}
  h \equiv h^q \mod (x^{q^2-1} - 1).
\end{equation}
Writing $h$ as $h(x)= \sum_{i=1}^{q+1} h_i x^i$, Equation~(\ref{eq:heqhq})
yields the system
\begin{equation}\label{syst:subfield}
\left\{
  \begin{array}{lcll}
    h_0 &= &h_0^q& \\ 
    h_1 &=& h_q^q& \\
    h_{q+1} &=& h_{q+1}^q& \\
    h_i & = &0,& \ \forall i \in \{2, \ldots, q-1\}.
  \end{array}
\right.
\end{equation}
Solving the above system yields the $\Fq$--basis of solutions:
$1, x+x^q, \alpha x + \alpha^q x^q, x^{q+1}$, which concludes the proof.
If $\xv$ is non full then it is easy to see that 
the polynomials satisfying (\ref{syst:subfield}) provide words of 
$\RS{q+1}{\xvzero} \cap \Fq^{n-1}$ but there might exist other ones.
\end{proof}

\bigbreak

\paragraph{\bf Discussion on the non full--support case}
In all our experiments, the code $\RS{q+1}{\xvzero} \cap \Fq^n$ 
turned out to have dimension $4$ even when the support is non full.
This can be explained as follows. In terms of polynomials, 
the full support code is generated as the image of the $\Fq$--space of 
polynomials in $\Fqq[z]_{<q+2}$ solution to the $\Fq$--linear system
\begin{equation}\label{eq:lin_forms}
 \forall x_i \in \Fqq^{\times},\quad f(x_i)^q - f(x_i) = 0.
\end{equation}
There are $q^2-1$ equations, while the $\Fq$--dimension
of $\Fqq[z]_{<q+2}$ is $2q+4$.
The non full support case is obtained by removing equations in 
(\ref{eq:lin_forms}). Since this system is overconstrained, 
one can reasonably hope that removing some equations will have
no incidence on the solution space as soon as $n > 2q+4$.

\bigbreak

\paragraph{\bf Conclusion} It is reasonable to hope --- and this is exactly what
happened in all our experiments (more than $600$ tests) --- that
$$\code{D} = \xvzero^{-(q+1)} \star \left(\RS{q+1}{\xvzero} \cap \Fq^{n-1}\right)
 = \xvzero^{-(q+1)} \star \EC.
$$

\subsubsection{The full weight codewords of $\DC$}
Since, the solution set of Problem~(\ref{eq:compNx}) consists in
full weight vectors $\cv$ with $c_1 = 1$, it is sufficient to classify
full weight vectors up to multiplication by a scalar.
For this reason, in what follows, we will frequently consider vectors
up to multiplication by a scalar.
According to the previous discussions, one can assume that $\code{D} =
\xvzero^{-(q+1)} \star \EC$ and the study of full weight codewords of 
$\code{D}$ reduces to that of $\EC$.

\begin{proposition}
  \label{prop:full_weight}
  Let 
  $$
  U \eqdef \left\{ \left.
       (\xvzero - a)^{q+1} ~\right|~ a\in \Fqq \setminus
      \{x_1, \ldots, x_{n-1}\}\right\}.
  $$
  Then, the elements of $U$ are full weight codewords of $\EC$.
  Moreover, let $\mathbf{P}$ be the probability that
  every full weight codeword of $\EC$
  up to multiplication by a scalar is in $U$, then
  $$
  \mathbf{P} \left\{
    \begin{aligned}
      =    &\ 1,  {\rm \ if\ } n \geq q^2 -q +2 \\
      \geq & 1-(q^3+q)\frac{{q^2-q+1 \choose n-1}}{{q^2 \choose n-1}} , {\rm \ else.}
    \end{aligned}
    \right.
  $$
\end{proposition}

\begin{proof}
Notice that the words $(\xvzero - a)^{q+1}$ for some
$a\in \Fqq$ are elements of $\EC$. Indeed, expanding
the word as  
$$
(\xvzero -a)^{q+1} = (\xvzero -a)^q(\xvzero -a) = \xvzero^{q+1} -\tr(a^q \xvzero) + a^{q+1}.
$$
Since $a^q$ decomposes as $a_0+\alpha a_1$, with $a_0, a_1 \in \Fq$,
this provides a decomposition of
$(\xvzero -a)^{q+1}$ as an $\F_q$--linear
combination of the words $\onev, \tr (\xvzero), \tr (\alpha \xvzero),
\xvzero^{q+1}$.

Let us investigate further the structure of the elements of $\EC$.
  The codewords of $\EC$ are given by evaluation
  of polynomials of the form
  \begin{equation}\label{eq:exprf}
  f(z) = \lambda_1 z^{q+1} + \lambda_2 (z^q+z) + \lambda_3 (\alpha^q z^q +
  \alpha z) + \lambda_4,\quad
  \lambda_1, \lambda_2, \lambda_3, \lambda_4 \in  \Fq.
  \end{equation}
  Therefore, describing the full weight codewords of $\EC$
  reduces to understand which of these polynomials do not vanish
  at any entry of $\xvzero$.
  Here, we can give a geometric interpretation of the set of roots in
  $\Fqq$ of such
  a polynomial in terms of points of affine conics over $\Fq$.
  For that we proceed to a Weil descent.
  Namely, set 
  $
  z = u+ \alpha v 
  $, where $u, v\in \Fq$.
  In addition, we choose $\alpha \in \Fqq \setminus \Fq$ so that
  $$
  \alpha^q = -\alpha,\quad \textrm{if}\ 2\nmid q
  \quad or \quad \alpha^q = \alpha +1 \quad \textrm{if}\ 2\mid q.
  $$
  Such an $\alpha$ always exists. Indeed, 
  \begin{itemize}
  \item in odd characteristic, choose a non-square $d\in \Fq$
    and let $\alpha \in \Fqq$ be a square root of $d$.
  \item in even characteristic, choose $d\in \Fq$ such that
    $\tr_{\Fq/\F_2}(d)\neq 0$, then the polynomial $z^2+z+d$
    is irreducible in $\Fq[z]$ and let $\alpha$ be one of
    its roots in $\Fqq$.
  \end{itemize}

Let us treat the odd characteristic case, the even characteristic can be 
treated in a very similar fashion.
Set $x=u+\alpha v$, where $\alpha \in \Fqq$
is a square root of a non-square $d \in \Fq$,  then
a simple computation from (\ref{eq:exprf}) transforms $f(x)$
as $\tilde{f}(u, v)$
\begin{equation}
  \label{eq:conic}
\tilde{f}(u, v) = \lambda_1 (u^2 -d v^2) + 2\lambda_2 u +2 d\lambda_3 v+\lambda_4,
\quad {\rm where} \quad \lambda_1,\lambda_2,\lambda_3,\lambda_4 \in \Fq.
\end{equation}
The set of pairs $(u, v) \in \Fq^2$ at which $\tilde{f}$
vanish are in one-to-one correspondence with the set of zeros in $\Fqq$ of
$f$. Therefore, $f$ provides a full-weight codeword in $\EC$
if and only if $f$ does not vanish on $\{x_1, \ldots, x_{n-1}\}$,
that is if and only if the zero locus of $\tilde{f}$ in $\Fq^2$
is contained in
$$A \eqdef \left\{ (u, v) ~\left|~ u+  \alpha v \in \Fqq \setminus \{x_1, \ldots, x_{n-1}\} \right. \right\} \cdot
$$
Consequently, we need
to understand the probability that $A$ contains a conic whose equation is
of the form (\ref{eq:conic}).
Let us analyze some particular cases
of conics of the form~(\ref{eq:conic}).
\begin{enumerate}[(i)]
  \item\label{it:empty} When $\lambda_1 = \lambda_2 = \lambda_3 = 0$. The conic
    is empty. In terms of codewords, it corresponds to multiples
    of the all-one word $\onev$.
  \item\label{it:lines}
    When, $\lambda_1 =0$ the conic is nothing but an affine line.
    It has exactly $q$ points over $\Fq$.
  \item\label{it:sing} For $\lambda_1 \neq 0$. Since we consider words only
    up to multiplication by a scalar, one can assume that $\lambda_1 = 1$.
    Let us look for a criterion for the conic to be singular. Recall that
    a conic of equation $f(x,y)$ is said to be {\em singular} if
    $f, \frac{\partial f}{ \partial x}$ and $\frac{\partial f}{\partial y}$
    have a common zero. The computation of the partial derivatives of
    $\tilde{f}$ yields:
    $$
    \frac{\partial \tilde{f}}{\partial u} = 2u+2\lambda_2 \quad
        \frac{\partial \tilde{f}}{\partial v} = -2dv+2d\lambda_3
    $$
    recall that we assumed $\lambda_1 = 1$. Then $\tilde{f}$ is singular
    if and only if $\tilde{f}(-\lambda_2, \lambda_3) = 0$ which is equivalent
    to
    $$
    \lambda_4 = \lambda_2^2 - d \lambda_3^2 =
    (\lambda_2 + \alpha \lambda_3)^{q+1}
    $$
    In such a situation a computation to $f$ from $\tilde{f}$ yields
   \begin{align*}
      f(z) & =  z^{q+1} + \lambda_2 \tr (z) + \lambda_3 \tr (\alpha z)
      +  (\lambda_2 + \alpha \lambda_3)^{q+1}\\
           & = (z - (\lambda_2 + \alpha \lambda_3))^{q+1}.
    \end{align*} 
    Therefore, the singular conics of the form~(\ref{eq:conic}) correspond to the words $(\xvzero -a)^{q+1}$
    for $a\in \Fqq$.
    In terms of codewords, either $a \in \{x_1, \ldots, x_n\}$ and the
    word  $(\xvzero -a)^{q+1}$ has weight $n-2$, or it is an element of $U$.
  \item\label{it:ellipse}
    Finally, for $\lambda_1 = 1$ and $\lambda_4 \neq \lambda_2^2 +
    d\lambda_3^2$, the conic is nonsingular and it is well-known
    that affine nonsingular conics have at least $q-1$ points (see for instance
    \cite[Chapter 9.3]{LN97}).
\end{enumerate}
In summary, only cases (\ref{it:lines}) and (\ref{it:ellipse})
may provide codewords of full weight which are not in $U$.
The number of lines coming from
(\ref{it:lines}) is the number of lines in the affine plane, namely $q^2+q$.
On the other hand, the number of conics coming from (\ref{it:ellipse})
is the number of possible triples $(\lambda_2, \lambda_3, \lambda_4)$
with $\lambda_4\neq \lambda_2^2 + d\lambda_3^2$.
That is $q^3-q^2$.

As a conclusion, 
there are $q^3 +q$ conics having at least 
$q-1$ points which may be contained in $A$.
The probability that such a conic is contained in $A$, which equals
the probability that the complement of $A$ is contained in the complement
of such a conic satisfies
$$
\prob\left(A\textrm{\ contains\ one\ of\ these}\ q^3+q\ \textrm{conics}\right)
\left\{
  \begin{array}{ll}
0 & {\rm if}\ |A| < q-1\\
\leq (q^3+q) \frac{{q^2-q+1 \choose q^2- |A|}}{{q^2 \choose q^2 - |A|}}&
{\rm else}.
  \end{array}
\right.
$$
Since $|A| = q^2 - (n-1)$, this yields the result.
\end{proof}

  \begin{table}[h]
    \centering
    
    \begin{tabular}{|c|c|c|c|c|}
      \hline
     $q= 29$,\ $n=791$ & $q=31$,\ $n=892$ & $q=31$, $n=851$ &
$q=31$,\ $n=813$ & $q=31$,\ $n=795$\\ 
 \hline
      $3\ 10^{-34}
$ & $2.3\ 10^{-33}$ & $4.7\ 10^{-26}$ & $1.06\ 10^{-21}$ & $4.7\ 10^{-20}$ \\
      \hline
    \end{tabular}
    \vspace{.2cm}
    \caption{Estimates of the upper bound on $1- \mathbf{P}$, where $\mathbf{P}$
      is defined in
      Proposition~\ref{prop:full_weight} for some explicit parameters.}
    \label{tab:probas}
  \end{table}

  \begin{rem}
    Actually, a further study proves that the nonsingular conics
    considered in the proof have all $q+1$ points. This permits
    to obtain a sharper bound for the probability. Details are left
    out here.
  \end{rem}

\subsubsection{Associating solutions by pairs}\label{ss:bypairs}
First remind that here again, words are considered only up to a
multiplication by a scalar.
In the previous subsections, we proved that with a very high probability,
the inverses of the solutions of Problem~(\ref{eq:compNx}) are
\begin{itemize}
\item $\xvzero^{q+1}$ which is the solution we look for;
\item the words $
\xvzero^{q+1} \star (\xvzero -a)^{-(q+1)}, \quad
a\in \Fqq \setminus \{x_0, \ldots, x_{n-1}\}.
$
\end{itemize}
In the very same manner, after computing the same filtration at position
$1$, one can then solve a linear problem of the form of
Problem~(\ref{eq:compNx}) whose inverse full weight solutions are
\begin{itemize}
\item $(\xvone - 1)^{q+1}$ which is the one we look for;
\item the words $
(\xvone - 1)^{q+1} \star (\xvone -a)^{-(q+1)} \qquad
a\in \Fqq \setminus \{x_0, \ldots, x_{n-1}\}.
$
\end{itemize}
Basically, we have two sets of $q^2 - n +1$ vectors (one can exclude
the all-one vector $\onev$ which is found easily).
The first set contains $\xvzero^{q+1}$ and the second one contains
$(\xvone - 1)^{q+1}$. But we do not know which ones they are.
The first idea would be to iterate Steps 3 and 4 of the attack until
the attack succeeds which represents in the worst case $(q^2- n+1)^2$
iterations. The point of this section is to explain how to 
reduce it to $q^2 -n +1$ iterations in the worst case.
For this purpose, let us bring in some notation.

\begin{nota}
Let $\xvzo$ be the vector $\xv$ punctured at positions $0,1$.
For all $a\in \Fqq \setminus \{x_0, \ldots, x_{n-1}\}$, set
$$
 \begin{array}{ccl}
  \uv_0(a) &\eqdef& 
  \spc{\xvzo^{q+1}}{(\xvzo-a)^{-(q+1)}} \\
  \uv_1(a) &\eqdef& 
  \spc{(\xvzo-1)^{q+1}}{(\xvzo-a)^{-(q+1)}}.
 \end{array}
$$
Moreover, set
$$
\uv_0(\infty)\eqdef \xvzo^{q+1} \quad
{\rm and}
\quad \uv_1(\infty)\eqdef(\xvzo-1)^{q+1},
$$
which can be regarded as $\uv_0(a)$ (resp. $\uv_1(a)$) ``when setting
$a=\infty$''. Finally, set 
\begin{align*}
\Lc_0 &\eqdef \{\uv_0(a) ~|~ a \in (\Fqq \cup \{\infty\})
 \setminus \{x_0, \ldots,
x_{n-1}\}\}\\
\Lc_1 &\eqdef \{\uv_1(a) ~|~ a \in (\Fqq \cup \{\infty\})
 \setminus \{x_0, \ldots,
x_{n-1}\}\}  
\end{align*}
\end{nota}

\begin{lemma}
  \label{lem:byPairs}
  Assume that $n>2q+4$.
  Let $a,b,c,d\in (\Fqq\cup\{\infty\}) \setminus \{x_0, \ldots, x_{n-1}\}$.
  Then, the vectors
  $\spc{\uv_0(a)}{\uv_1(b)}$ and $\spc{\uv_0(c)}{\uv_1(d)}$ are collinear
  if and only if
  $$
  \begin{array}{rccc}
    {\rm either} & a  =  c & {\rm and} & b  =  d \\
    {\rm or}     & a  =  d & {\rm and} & b  =  c.
  \end{array}
  $$
\end{lemma}

\begin{proof}
  The ``if'' part is straightforward.
  Conversely, assume that $\spc{\uv_0(a)}{\uv_1(b)}$ and $\spc{\uv_0(c)}{\uv_1(d)}$ are collinear. Thus, there exists a nonzero scalar $\lambda \in \Fqq$ such that
  \begin{equation}
    \label{eq:collinearity}
  \spc{\uv_0(a)}{\uv_1(b)} =\lambda \spc{\uv_0(c)}{\uv_1(d)}.  
  \end{equation}
For convenience, we assume that $a,b,c$ and $d$ are all distinct from $\infty$.
The cases when some of them equal $\infty$ are treated in the same way-
we therefore omit to detail these cases here.
From (\ref{eq:collinearity}), we have
$$
\forall i \in \{2, \ldots , n-1\},\ \
{\left(\frac{x_i (x_i-1)}{(x_i-a)(x_i-b)}
\right)}^{q+1}= \lambda 
{\left(\frac{x_i (x_i-1)}{(x_i-c) (x_i-d)}
\right)}^{q+1}.
$$
This leads to
\begin{equation}
  \label{eq:coll2}
\forall i \in \{2, \ldots , n-1\},\ \ 
  (x_i-c)^{q+1} (x_i-d)^{q+1} = \lambda (x_i-a)^{q+1} (x_i-b)^{q+1},
\end{equation}
From (\ref{eq:coll2}), the polynomial $P(z)\eqdef
((z-c)(z-d))^{q+1} - \lambda ((z-a)(z-b))^{q+1}$ vanishes at $x_i$
for all $i \in \{2, \ldots , n-1\}$, and hence has $n-2$ roots, while its degree
is less than or equal to $2q+2$. Thus, under the assumption $n>2q+4$, this polynomial has more roots than its degree and hence is zero. 
This yields the result.
\end{proof}

\begin{proposition}\label{prop:by_pairs}
  Assume that $n > 2q+4$
  Let $a, a'\in (\Fqq \cup \{\infty\}) \setminus \{x_0, \ldots, x_{n-1}\}$.
  If we have the following equality of sets:
  $$
  \{\uv_0(a) \star \cv ~|~ \cv \in \Lc_1\} = \{\cv' \star \uv_1 (a') ~|~
  \cv' \in \Lc_0\},
  $$
  where vectors are considered up to multiplication by a scalar,
  then, $a=a'$.
\end{proposition}

\begin{proof}
  Clearly, if $a = a'$
  then every element of the left hand set is of the form
  $\uv_0(a)\star \uv_1(b)$, for some
  $b\in  (\Fqq \cup \{\infty\}) \setminus \{x_0, \ldots, x_{n-1}\}$
  and, from Lemma~\ref{lem:byPairs}, this vector is collinear to
  $\uv_0 (b) \star \uv_1 (a)$.
  
  Now, if $a \neq a'$, then let $b \in (\Fqq \cup \{\infty\})
  \setminus \{x_0, \ldots, x_{n-1}\}$ and $b \neq a, a'$. Then,
  Lemma~\ref{lem:byPairs} asserts that for all $c\in (\Fqq  \cup \{\infty\})
  \setminus \{0, \ldots, n-1\}$,
  $\uv_0(c) \star \uv_1(a')$ is non collinear to $\uv_0(a)\star \uv_1(b)$.
\end{proof}

Proposition~\ref{prop:by_pairs} allows to gather elements of $\Lc_0, \Lc_1$
by pairs $(\uv_0(a), \uv_1(a))$ without knowing $a$. We proceed as follows:
we compute all the sets
$$
\av_0 \star \Lc_1 \eqdef \{\av_0 \star \cv ~|~ \cv \in \Lc_1\}
$$
for all $\av_0 \in \Lc_0$
and all the sets
$$
\Lc_0 \star \av_1 \eqdef \{\cv' \star \av_1 ~|~ \cv' \in \Lc_0\}.
$$
for all $\av_1 \in \Lc_1$.

Then, if two such sets match i.e. if $\av_0 \star \Lc_1 = \Lc_0 \star \av_1$,
then we create the pair $(\av_0, \av_1)$. 
By Proposition ~\ref{prop:by_pairs} they correspond to pairs of the form $(\uv_0(a), \uv_1(a))$.
By this manner, we create $q^2-n+1$ pairs 
of elements of $\Lc_0 \times \Lc_1$. One of them is the one we look for,
namely the pair $(\xv^{q+1}, (\xv - 1)^{q+1})$.

\subsection{Further details on Step 4 of the attack}
\label{ss:further_step4}

Thanks to Lemma~\ref{lem:MiniPol}, one can compute the minimal polynomial
$P_{x_i}$ of every entry $x_i$ of $\xv$, that is to say that the support
is known up to Galois action.

\begin{fact}
  \label{fact:sigma}
From the very knowledge of these $P_{x_i}$'s, one can compute a (non unique) permutation $\sigma$ such that
$\xsig\eqdef \sigma(\xv)$ is of the form
\begin{equation}\label{eq:xsigma}
\xsig = (u_0, u_1, \ldots , u_{\ell_1}, v_0, v_0^q, \ldots, v_{\ell_2-1}, v_{\ell_2-1}^q, w_0, \ldots , w_{\ell_3-1})
\end{equation}
where
\begin{itemize}
\item the $u_i$'s list all the entries of $\xv$ lying in $\Fq$;
\item the $v_i$'s list all the entries of $\xv$  in $\Fqq\setminus \Fq$
  and whose conjugate is also an entry of $\xv$.
\item the $w_i$'s list all the entries of $\xv$ in $\Fqq\setminus \Fq$
  and whose conjugate is not an entry of $\xv$.
\end{itemize}
\end{fact}

Therefore, one can compute a
generator matrix of the code $\CCs \eqdef \Goppa{\xsig}{\gamma^{q+1}}$
by permuting the columns of a generator matrix of $\CC$.
Call $\Gsig \in \Fq^{k\times n}$ this matrix.
In other words $\Gsig$ is obtaining by first picking 
the columns of a generator matrix $\Gm$ of $\CC$ that correspond to the 
entries of $\xv$ that belong to $\fq$ and then putting together the 
columns of $\Gm$ that correspond to conjugate entries of $\xv$ and finally put at the end
the columns of $\Gm$ that are not of this kind.

It is worth noting that, even when $\sigma$ is computed,
the vector $\xsig$ remains unknown since only the minimal polynomials
of its entries are known. 
Afterwards, we introduce an extended support by inserting the conjugates of the
$w_j$'s.
\begin{equation}\label{eq:xve}
\xve \eqdef (u_0, \ldots , u_{\ell_1-1},
v_0, v_0^q\ldots , v_{\ell_2-1}, v_{\ell_2-1}^q, w_0, w_0^q,\ldots , w_{\ell_3-1}, w_{\ell_3-1}^q).
\end{equation}
This vector is also unknown, however, one can compute an arbitrary vector
which equals $\xve$ up to a very particular permutation.
One can compute an arbitrary vector
\begin{equation}\label{eq:xvpe}
\xvpe =  (u_0, \ldots , u_{\ell_1-1}, v'_0, (v'_0)^q,\ldots
{v'_{\ell_2-1}}, (v'_{\ell_2-1})^q,
 w'_0, {w'_0}^q,\ldots
{w'_{\ell_3-1}}, (w'_{\ell_3-1})^q)
\end{equation}
such that for all $i$ the $i$-th entry of $\xvpe$ has the same
minimal polynomial as that of $\xve$. Equivalently, the entries
of $\xvpe$ equal those of $\xve$ up to Galois action.
This can be interpreted in terms of permutations using:
\begin{definition}
  \label{def:Tau}
Let $\mathfrak{T}$ be the subgroup of $\mathfrak{S}_{n+\ell_3}$ of products
of transpositions with disjoint supports, each one permuting
either the positions $v_i, v_i^q$ the positions
 $w_i, w_i^q$ in $\xve$.
Every element $\tau \in \mathfrak{T}$ is represented by the matrix
\begin{equation}\label{eq:Rtau}
\mathbf{R}_{\tau} \eqdef \begin{pmatrix}
  \mathbf{I}_{\ell_1} &  (0)\\
   (0) & \mathbf{B}  \\
\end{pmatrix},
\end{equation}
where $\mathbf{B} \in \mathfrak{M}_{2(\ell_2+\ell_3)}(\Fq)$
is $2\times 2$ block--diagonal 
with blocks among $
\begin{pmatrix}
  1 & 0 \\ 0 & 1
\end{pmatrix}
$
and
$
\begin{pmatrix}
  0 & 1 \\ 1 & 0
\end{pmatrix}
$.
\end{definition}

\begin{lemma}
  \label{lem:permutxvpe}
There exists $\tau \in \mathfrak{T}$ such that $\tau (\xve) = \xvpe$.
\end{lemma}

We extend $\Gsig$
as a $k\times (n+\ell_3)$ matrix $\Gsige$ by inserting $\ell_3$ zero columns
at positions in one-to-one correspondence with the entries $w_i^q$ in $\xve$
(see (\ref{eq:xve})).
That is:
\begin{equation}\label{eq:Gsige}
 \Gsig = \begin{pmatrix}
 g_{11} & \ldots& g_{1n}\\
 \vdots &       & \vdots\\
 g_{k1} &\ldots & g_{kn}
 \end{pmatrix}
\quad {\rm and}
\quad
\Gsige = \begin{pmatrix}
 g_{11} & \ldots& g_{1s} & 0 & g_{1,s+1} & 0 & \ldots & g_{1n} & 0\\
 \vdots &       &        &   &           &   &        &        & \vdots\\
 g_{k1} &\ldots & g_{ks} & 0 & g_{k,s+1} & 0 & \ldots & g_{kn} & 0
 \end{pmatrix},
\end{equation}
where $s = \ell_1+2\ell_2$ (see (\ref{eq:xsigma})).
The corresponding code is referred to as $\CCse$.
The key of this final step is the following statement.

\begin{theorem}
  \label{thm:finalkey}
There exists a matrix $\mathbf{M}$ such that
\begin{equation}\label{eq:finalkey}
\CCse\ \mathbf{M}  ~\subseteq~ \Alt{r(q+1)}{\xvpe}{\onev} 
\end{equation}
and $\mathbf{M} = \mathbf{R}_{\tau}\mathbf{D}$, where $\mathbf{D}$ is diagonal and invertible and $\mathbf{R}_{\tau}$ is the permutation matrix of some $\tau \in \mathfrak{T}$ as described in (\ref{eq:Rtau}).
\end{theorem}

\begin{proof}
  Recall that $\CCs = \Goppa{\xsig}{\gamma^{q+1}}$.
Since from Definition~\ref{def:ClassicalGoppa},
Goppa codes are alternant and hence
from Proposition~\ref{pr:shortened_alternant_code},
the code $\CCs$ is a shortening of $ \Goppa{\xve}{\gamma^{q+1}}$.
Consequently, we have
\begin{equation}
  \label{eq:incCCse}
\CCse \subseteq \Goppa{\xve}{\gamma^{q+1}}.  
\end{equation}
Second,
from Lemma~\ref{lem:permutxvpe}, there exists $\tau \in \mathfrak{T}$
such that $\tau(\xve) = \xvpe$, then
\begin{equation}
  \label{eq:PermAlt}
 \Goppa{\xve}{\gamma^{q+1}} \mathbf{R}_{\tau}  = \Goppa{\xvpe}{\gamma^{q+1}} ,  
\end{equation}
Next,
from Theorem~\ref{thm:properties_superGoppa}(\ref{it:PsupSRS}),
there exists $\av \in {(\Fq^{\times})}^{n+\ell_3}$  such that
\begin{equation}
  \label{eq:super}
\Goppa{\xvpe}{\gamma^{q+1}} = \spc{\av}{\Alt{r(q+1)}{\xvpe}{\onev}}.  
\end{equation}
Let $\mathbf{D}$ be the diagonal matrix whose diagonal equals $\av^{-1}$. From (\ref{eq:incCCse}), (\ref{eq:PermAlt}) and (\ref{eq:super}),
we get
\begin{equation}
 \CCse \mathbf{R}_{\tau}\mathbf{D} \subseteq \Alt{r(q+1)}{\xvpe}{\onev}.  
\end{equation}
\end{proof}

To finish the attack, we proceed as follows. We compute a permutation $\sigma$ as in Fact \ref{fact:sigma}. Then, we compute the matrix $\Gsige$
defined in (\ref{eq:Gsige}).
Afterwards, we compute an arbitrary vector $\xvpe$ 
and a parity--check matrix $\mathbf{H}$ of the code
$\Alt{r(q+1)}{\xvpe}{\onev}$.
 Finally,
we solve the problem
\begin{problem}
  \label{prob:final}
  Find the space of matrices $\mathbf{M}\in \mathfrak{M}_{n+\ell_3}(\Fq)$
  of the form
  $
  \mathbf{M} = \begin{pmatrix}
    \mathbf{E} &  (0)\\
     (0) & \mathbf{F} 
  \end{pmatrix}
  $, where $\mathbf{E}$ is $\ell_1\times \ell_1$ and diagonal
  and $\mathbf{F}$ is $(2\ell_2+2\ell_3)\times (2\ell_2+2\ell_3)$
  and $2\times 2$--block--diagonal such that
  $$
  \mathbf{H}\ {\left( \Gsige\ \mathbf{M}\right)}^{T} ~=~ 0.
  $$
\end{problem}

The matrix $\mathbf{M}$ of Theorem~\ref{thm:finalkey} is a solution of
Problem~\ref{prob:final}.
Moreover, this problem is linear, has  $\ell_1 + 2(\ell_2+\ell_3) \leq 4n$
unknowns and $(\dim \CC)(n-\dim \Alt{r(q+1)}{\xvpe}{\onev})\geq k(n-k)$
equations. Thus, the number of unknowns is linear while the number of equations
is quadratic. This provides an extremely small space of solutions.

\begin{example}
  If we consider a $[841, 601]$ wild Goppa code over $\F_{32}$
  (where $r=4$), then we get less than
  $3364$ unknowns and more than $120200$ equations.
\end{example}

Experimentally, we observe that the solution space has dimension $2$
and an exhaustive search of matrices which are the product of a diagonal matrix
and a permutation matrix provide two solutions (see Lemma~\ref{lem:conjugsupp}
for the rationale behind these two solutions). Choose a solution
$\mathbf{M}$, then factorize it as $\mathbf{D}\mathbf{R}_{\tau}$
as in Theorem~\ref{thm:finalkey}.
This yields the permutation $\tau$ and hence the support
$\xv$. Second, the entries of $\mathbf{D}$ provide directly a vector $\av$
such that 
\begin{equation}
  \label{eq:victory}
\CC = \spc{\av}{\Alt{r(q+1)}{\xv}{\onev}},  
\end{equation}
 which allows
to correct up to $\lfloor \frac{r(q+1)}{2}\rfloor$ errors.
Hence the scheme is broken.

\medbreak
\begin{algorithm}
 \caption{\label{al:total} Algorithm of the attack.}
\begin{algorithmic}
\STATE{Compute $\Coi{0}{q+1}$, $\Coi{1}{q+1}$ using 
Algorithm \ref{al:sCoi0v}.}
 \STATE{$ \Lc_0 \leftarrow$ List of candidates for $\xvzero^{q+1}$ (Obtained by solving (\ref{eq:compNx}))}
 \STATE{$ \Lc_1 \leftarrow$ List of candidates for $(\xvone-1)^{q+1}$}
 \STATE{$\mathcal{P} \leftarrow$ the set of $q^2-n+1$ pairs $(\av_0, \av_1)\in\Lc_0 \times \Lc_1$ as explained in \S\ref{ss:bypairs}.}
 \STATE{$\mat{M}_0 \leftarrow 0$}
 \WHILE{$\mat{M}_0=0$ and $L_0 \neq \emptyset$}
 \STATE{$(\av_0, \av_1) \leftarrow$ a random pair in $\mathcal{P}$.}
 \STATE{$\mathcal{P} \leftarrow \mathcal{P}\setminus \{(\av_0, \av_1)\}$}
 \STATE{ Compute the minimal polynomials $P_i$ of the positions using Lemma~\ref{lem:MiniPol}.}
 \STATE{Construct $\Gsige$, $\xvpe$ and
 a parity-check matrix $\mat{H}$ of the code $\Alt{r(q+1)}{\xvpe}{\onev}$
 as described in Theorem~\ref{thm:finalkey}.}
 \STATE{$V \leftarrow$ Space of solutions $\mat{M}$ of
   Problem~\ref{prob:final}}
 \IF{$\dim V>0$ and $\exists \mat{M} \in V$ of the form
 $\mathbf{D}\mathbf{R}_{\tau}$ as in Theorem~\ref{thm:finalkey}}
 \STATE{$\mat{M}_0 \leftarrow M$}
 \ENDIF
 \ENDWHILE
 \IF{$\mat{M}_0=0$}
 \RETURN{``error''}
 \ELSE
 \STATE{Recover $\xv$ and $\uv$ from $\mat{M}$ as in (\ref{eq:victory})}
 \RETURN{$\xv, \uv$}
 \ENDIF
\end{algorithmic}
\end{algorithm}

\paragraph{\bf The two solutions of the problem}
The solutions of 
Problem~\ref{prob:final} of the form $\mathbf{D}\mathbf{R}_{\tau}$,
where $\mathbf{D}$ is diagonal and invertible and $\mathbf{R}_{\tau}$
is a permutation matrix has cardinality $2$. This is due to the fact
that any alternant code of extension degree $2$ has at least $2$
pairs $(\xv, \yv)$ to represent it.
This explained in the following lemma.

\begin{lem}
  \label{lem:conjugsupp}
Let $\av \in \Fqq^n$ be a support and $\bv \in \Fqq^n$ be a multiplier.
Then,
$$
\GRS{k}{\av}{\bv}\cap \Fq^n = \GRS{k}{\av^q}{\bv^q} \cap \Fq^n.
$$
\end{lem}

\begin{proof}
Let $f\in \Fqq[x]_{<k}$ be a polynomial such that $(b_0 f(a_0),\ldots , b_{n-1}f(a_{n-1})) \in \GRS{k}{\av}{\bv}\cap \Fq^n$.
Writing $f$ as $f_0+f_1x+\cdots + f_{k-1}x^{k-1}$, denote by $f^{(q)}\in \Fqq[x]_{<k}$ the polynomial $f^{(q)} \eqdef f_0^q+ f_1^q x+\cdots + f_{k-1}^q x^{k-1}$. Then it is easy to check that for all $i \in \{0, \ldots , n-1\}$, we have
$$
b_i^q f^{(q)} (a_i^q) = (b_i f(a_i))^q.
$$
In addition, since by assumption $b_i f(a_i) \in \Fq$, we have
$$
\forall i\in \{0, \ldots , n-1\},\ b_i^q f^{(q)} (a_i^q) = b_i f(a_i)
$$
Therefore,
$$
(b_0 f(a_0), \ldots , b_{n-1}f(a_{n-1})) = (b_0^q f^{(q)}(a_0^q), \ldots , b_{n-1} f^{(q)}(a_{ n-1}^q)) \in \GRS{k}{\av^q}{\bv^q}\cap \Fq^n. 
$$
We proved that $\GRS{k}{\av}{\bv}\cap \Fq^n \subseteq \GRS{k}{\av^q}{\bv^q} \cap \Fq^n$ and the converse inclusion can be proved by the very same manner.
\end{proof}


\end{document}